\documentclass[12pt]{article}

\usepackage[letterpaper,
            top=1.05in,
            bottom=1.05in,
            left=1.10in,
            right=1.10in]{geometry}

\usepackage{setspace}
\usepackage{amsmath,amsthm,amssymb,mathtools,bm}
\usepackage{newtxtext,newtxmath}
\usepackage{mathrsfs}

\allowdisplaybreaks

\usepackage{graphicx}
\usepackage{booktabs}
\usepackage{caption}
\usepackage{tikz}
\usetikzlibrary{
    arrows.meta,
    decorations.pathreplacing,
    positioning,
    calc
}

\usepackage{enumitem}
\setlist[itemize]{topsep=3pt,itemsep=2pt,parsep=0pt}
\setlist[enumerate]{topsep=3pt,itemsep=2pt,parsep=0pt}

\usepackage{titlesec}

\titleformat{\section}
  {\large\bfseries}
  {\thesection.}{0.65em}{}

\titleformat{\subsection}
  {\normalsize\bfseries}
  {\thesubsection.}{0.65em}{}

\titleformat{\subsubsection}
  {\normalsize\itshape}
  {\thesubsubsection.}{0.65em}{}

\titlespacing*{\section}
  {0pt}{2.2ex plus .5ex minus .2ex}{1.0ex}

\titlespacing*{\subsection}
  {0pt}{1.8ex plus .4ex minus .2ex}{0.7ex}

\titlespacing*{\subsubsection}
  {0pt}{1.5ex plus .3ex minus .2ex}{0.5ex}

\newtheorem{thm}{Theorem}
\newtheorem{prop}{Proposition}
\newtheorem{lem}{Lemma}
\newtheorem{cro}{Corollary}

\theoremstyle{definition}
\newtheorem{defn}{Definition}

\newtheorem*{exam*}{Example}

\theoremstyle{remark}
\newtheorem{rmk}{Remark}

\usepackage[round,authoryear]{natbib}
\usepackage{xcolor}
\definecolor{linkblue}{RGB}{25,70,120}

\usepackage{hyperref}
\hypersetup{
    colorlinks=true,
    citecolor=linkblue,
    linkcolor=linkblue,
    urlcolor=linkblue,
    bookmarksopen=true,
    pdfstartview={XYZ null null 1.00}
}

\usepackage[nameinlink,capitalize,noabbrev]{cleveref}

\crefname{thm}{theorem}{theorems}
\Crefname{thm}{Theorem}{Theorems}
\crefname{prop}{proposition}{propositions}
\Crefname{prop}{Proposition}{Propositions}
\crefname{lem}{lemma}{lemmas}
\Crefname{lem}{Lemma}{Lemmas}
\crefname{cro}{corollary}{corollaries}
\Crefname{cro}{Corollary}{Corollaries}
\crefname{defn}{definition}{definitions}
\Crefname{defn}{Definition}{Definitions}
\crefname{assum}{assumption}{assumptions}
\Crefname{assum}{Assumption}{Assumptions}
\crefname{rmk}{remark}{remarks}
\Crefname{rmk}{Remark}{Remarks}

\newcommand{\E}{\mathbb{E}}

\newcommand{\R}{\mathbb{R}}

\newcommand{\cA}{\mathcal{A}}

\newcommand{\cD}{\mathcal{D}}
\newcommand{\cF}{\mathcal{F}}

\newcommand{\cM}{\mathcal{M}}

\newcommand{\dd}{\,d}

\DeclareMathOperator{\Var}{Var}

\title{\Large\bfseries Delegation without Priors\thanks{
Wei He: The Chinese University of Hong Kong (email: hewei@cuhk.edu.hk). Fei Li: University of North Carolina at Chapel Hill (email: lifei@email.unc.edu). Jiangtao Li: The Hong Kong University of Science and Technology (email: jiangtaoli@ust.hk). Terrence McGovern: University of Bonn (email: mcgovern@uni-bonn.de). Jan Benedikt Napp: University of Bonn (email: jbnapp@uni-bonn.de). Hanqing Ye: University of North Carolina at Chapel Hill (email: hanqing@unc.edu). Wenjun Zheng: University of Bonn (email: wenjun@uni-bonn.de). We thank Ricardo Alonso, Sarah Auster, Roberto Corrao, Songzi Du, Tan Gan, Yingni Guo, Paul Heidhues, Andreas Kleiner, Daniel Kr\"ahmer, Yingkai Li, Qingmin Liu, Benny Moldovanu, Dezső Szalay, Boli Xu, Kai Hao Yang, and audiences at the UNC theory workshop and the Yale Theory Breakfast for valuable comments and insightful discussions. This paper subsumes two independently developed papers: %\cite{he2026delegation} 
\cite{li2026delegation} 
and \cite{mcgovern2026prior}.
}}

\author{%
\begin{tabular}{c}
Wei He \quad Fei Li \quad Jiangtao Li  \quad Terrence McGovern \\[0.6em]
Jan Benedikt Napp \hspace{1.2em} Hanqing Ye \hspace{1.2em} Wenjun Zheng
\end{tabular}
}

\date{\today}

\begin{document}

\maketitle

\thispagestyle{empty}

\begin{abstract}

A principal makes decisions relying on the private information held by a biased agent. Lacking a prior over the circumstances that may arise, she seeks a delegation rule that minimizes worst-case regret relative to the complete-information benchmark. The optimal delegation mechanism governs two distinct margins: whether to depart from a common default action, and how much discretion to grant conditional on departure. This rule pairs the default action with an action gap, followed by a discretionary interval with randomized boundaries designed to balance regret across
states. This structure is globally optimal within the class of all stochastic, incentive-compatible mechanisms.

\end{abstract}

\vspace{0.5em}

%\noindent

\medskip
%\noindent

\bigskip

\newpage

%\pagenumbering{arabic} 
\setcounter{page}{1}

\section{Introduction}\label{sec:introduction}

In many organizations, a principal delegates decision-making authority to a privately informed agent. Delegation puts the agent's information to use, but it gives the agent room to steer the decision toward his own interests. The central design problem is therefore how much discretion to grant the agent. Restricting the agent's discretion can limit his ability to choose against the principal's interests, but it also makes the decision less responsive to the circumstances at hand. The optimal scope of delegation therefore depends on which circumstances the principal expects to arise more frequently.

Forming such an assessment can be difficult. Consider a firm extending an established line of business into a new market. While headquarters retains authority over investment, a local manager observes current market conditions and knows better whether further expansion is warranted and how much to invest. The manager's information makes delegation valuable, but an empire-building motive may lead him to favor larger investments than headquarters would prefer.\footnote{Practitioner accounts of capital budgeting explicitly discuss project sponsors' incentives to build larger businesses, which may bring higher compensation and promotion opportunities; see \citet[pp.~185--186]{barshop2016capital}.} With limited experience in the new market, headquarters may lack a reliable basis for assessing how frequently different market conditions will arise. 

We study delegation when the principal lacks a reliable prior over the circumstances that may arise. The principal's decision is modeled as a real number, whose desirability depends on a one-dimensional payoff-relevant state. The agent observes the realized state. We adopt the canonical quadratic-loss environment \citep{holmstrom1980theory, melumad1991communication} in which both parties prefer higher actions as the state increases, while the agent's preferred action is shifted upward relative to the principal's. Monetary transfers are unavailable.

Before the state is realized, the principal chooses how much discretion to grant the agent. We allow her to use random delegation rules, with the agent free to choose an action from the realized delegation set. The principal evaluates a rule by its worst-case performance. Specifically, at each state, the principal incurs a loss due to the mismatch between her ideal action and the agent's choice. Since the principal could achieve zero loss by choosing her ideal action under full information, this loss also measures regret relative to the full-information benchmark. The principal chooses a delegation rule to minimize the largest statewise regret. The principal need not literally expect the worst case to occur. The criterion instead provides a prior-free guarantee: a rule that performs well under it does so regardless of which state arises.

We characterize the optimal delegation rule and also show that it attains the optimum among all stochastic incentive-compatible mechanisms. This rule has a simple structure. A common low default action is always available, and a fixed range of actions immediately above it is always excluded. Beyond this gap lies a discretionary interval. All randomization is summarized by two variables: the floor and the ceiling of this discretionary interval. Figure \ref{fig:randomized-delegation} illustrates this structure.

\begin{figure}[h]
\centering
\begin{tikzpicture}[x=1cm,y=1cm,>=stealth]

% ============================================================
% Left panel
% ============================================================
\begin{scope}[shift={(0,0)}]

%\node[font=\small] at (3.1,1.05) {\footnotesize Delegation Set 1};

% Action space = R
\draw[<->, line width=0.7pt] (0,0) -- (6.2,0);

% Common default
\fill [black!50](0.9,0) circle (2.7pt);

% Discretionary interval
\draw [black!50, line width=4pt, line cap=round]
    (3.25,0) -- (5.15,0);

% Labels
\node[font=\scriptsize, align=center]
    at (0.9,-0.3) {default};

\node[font=\scriptsize]
    at (2.1,-0.30) {gap};

\node[font=\scriptsize]
    at (4.2,-0.30) {discretionary interval};
    % Common gap marker
\draw[black!45, thin] (1.05,0.20) -- (3.00,0.20);
\draw[black!45, thin] (1.05,0.14) -- (1.05,0.26);
\draw[black!45, thin] (3.00,0.14) -- (3.00,0.26);

\end{scope}

% ============================================================
% Right panel
% ============================================================
\begin{scope}[shift={(9.0,0)}]

%\node[font=\small] at (3.1,1.05) {\footnotesize Delegation Set 2};

% Action space = R
\draw[<->, line width=0.7pt] (0,0) -- (6.2,0);

% Same common default
\fill [black!50](0.9,0) circle (2.7pt);

% Different discretionary interval
\draw[black!50, line width=4pt, line cap=round]
    (3.85,0) -- (5.70,0);

% Labels
\node[font=\scriptsize, align=center]
    at (0.9,-0.3) {default};

\node[font=\scriptsize]
    %at (2.4,-0.30) {gap};
        at (2.1,-0.30) {gap};

\node[font=\scriptsize]
    at (4.78,-0.30) {discretionary interval};

% Common gap marker
\draw[black!45, thin] (1.05,0.20) -- (3.00,0.20);
\draw[black!45, thin] (1.05,0.14) -- (1.05,0.26);
\draw[black!45, thin] (3.00,0.14) -- (3.00,0.26);

\end{scope}

\end{tikzpicture}

\caption{
Two realized delegation sets of the optimal randomized delegation rule.}
\label{fig:randomized-delegation}
\end{figure}

The optimal delegation mechanism exploits the partial alignment of preferences. The default and the gap create a hurdle for departing from the baseline. At low states, the agent would like to move somewhat upward, but the actions above the gap are too high for him to prefer them to the default. As the state rises, his ideal action rises as well, and eventually he is willing to cross the gap and use the higher discretionary interval. The discretionary interval then allows him to adjust the decision to his information. Its floor determines the hurdle for leaving the default, while its ceiling limits how far the decision can increase. Randomizing these boundaries varies both the hurdle for leaving the default and the scope of discretion once it is crossed, allowing the principal to hedge her exposure across states.

This default-gap-interval structure offers a broader view of familiar delegation arrangements. Many applications naturally feature a default together with a range of alternatives that becomes available only after a meaningful hurdle is crossed. A manager may maintain a routine investment level, while a major expansion opens access to a higher range of investment scales. In academic hiring, a dean may authorize a department to make a standard salary offer, while an exceptional candidate may justify approval of an offer in a substantially higher range. These examples share a common feature: discretion over a range of actions may be granted only when the case for departing from the default is sufficiently strong. Delegation therefore governs two related margins: the threshold for abandoning the default and the latitude afforded to the agent once that threshold is met, which are captured by the gap and the higher interval in our optimal rule.

To establish optimality, we first solve the unrestricted mechanism design problem and then show that our construction implements the optimal mechanism. %Incentive compatibility sharply constrains how a mechanism can respond to the agent's information. 
The key tradeoff is between responsiveness and action dispersion. Responsiveness requires the average action to increase with the reported state. As long as this average remains below the agent's preferred action, higher reports are more attractive to lower types. Incentive compatibility then requires the dispersion of the induced action to increase as well. The only way to avoid this additional dispersion while continuing to raise the average action is to give the agent his preferred action on average. Under worst-case evaluation, doing so leaves the principal fully exposed to the preference conflict in any state, which is undesirable. The optimal robust mechanism therefore uses increasing dispersion whenever it remains responsive to the agent's information.

This logic also shapes where responsiveness begins and ends. At the lowest state, dispersion provides no incentive benefit and only adds to the principal's loss, so the mechanism assigns a deterministic action. As the reported state rises, the mean action increases, and the dispersion required for incentive compatibility accumulates. Eventually, further responsiveness would require too much additional dispersion. The mechanism then stops increasing the average action, which also allows dispersion to stop rising, and pools the remaining states.

A random delegation rule implements the optimal mechanism. The deterministic outcome at the lowest state generates the common default and gap, while randomizing the boundaries of the higher interval reproduces the responsiveness and dispersion of the unrestricted optimum. Thus, despite the stronger ex post incentive requirement imposed by delegation, the principal loses nothing relative to unrestricted stochastic mechanism design.

The discussion above highlights the role of randomization in robust delegation. Without randomization, any fixed delegation set induces a deterministic action in every state and therefore generates none of the action dispersion required by the optimal mechanism. The limitation is particularly stark under worst-case evaluation. Once an action is included in the delegation set, it is always available to the agent. There may then be a state in which this action is exactly what the biased agent wants to choose while being far from the principal's preferred action, leaving the principal fully exposed to the resulting loss of control. Randomizing the delegation set makes access to such actions probabilistic, generating the dispersion needed for incentives while limiting the principal's exposure to any particular circumstance.

With a constant bias, preserving the common default at low states requires excluding actions immediately above it. At the lowest state, the optimal mechanism prescribes the same deterministic action in every realization, which becomes the common default under delegation. Because the lowest-state agent would prefer a range of actions immediately above this default, all of these actions must be excluded from every realized delegation set, creating a common gap. Other states, however, require access to actions above this gap. This is where random interval delegation fails. Any realized interval that contains both the common default and a sufficiently high action must also contain every action between them, thereby filling in the gap. Randomizing over intervals changes which connected set is available, but cannot reproduce the disconnected choice set required by the optimal mechanism.

Finally, we consider two extensions. First, the principal is uncertain about both the state and the magnitude of the agent's bias, knowing only an upper bound on the latter. We show that the optimal rule from the constant-bias model remains minimax optimal, even when the state and the bias are arbitrarily dependent. Second, we allow the agent's bias to vary with the state. The principal knows how the agent's bias depends on the state but remains uncertain about the state itself. With affine bias, sufficiently large disagreement at low states calls for a single default protected by a gap as in the baseline model. When low-state disagreement is small, the default can instead be an interval that grants limited discretion, with a random cap determining how much additional discretion to provide. A similar pattern emerges under general nonnegative monotone biases: it remains optimal to guarantee a common set of default actions while randomizing additional discretion, although the latter may take more complex forms.

\paragraph{Related Literature.} Starting from \citet{holmstrom1980theory}, there is a large literature on optimal delegation under known priors \citep{melumad1991communication, kovac2009stochastic, martimort2006continuity,
alonso2008optimal,amador2013theory,amador2025generalized}.  
With quadratic loss and a constant bias, this literature shows that under a regularity condition on the prior, interval delegation is optimal even within the class of stochastic mechanisms \citep{kovac2009stochastic}. When the prior is unknown, the answer changes qualitatively. The optimal rule protects a common default by excluding actions directly above it, and grants discretion only beyond this gap: stochastically, over an interval with random bounds. The optimal robust rule thus separates access to discretion from its scope.

Non-convex delegation sets also arise in other settings with known priors. In models with hidden information acquisition or effort, excluding intermediate actions can strengthen the agent's incentives to become informed or exert effort \citep{szalay2005economics,ball2024benefiting,gerardi2024delegation}. In models with state-dependent preference misalignment \citep{melumad1991communication,alonso2008optimal}, a gap induces more extreme choices when the agent's preferred action responds less to the state than the principal's. Despite these different foundations, gaps play a similar role: they remove attractive intermediate responses and make behavior more responsive to underlying circumstances. We highlight a complementary role for gaps under a directional upward constant bias, where the principal's and agent's preferred actions have the same state sensitivity but the agent's is systematically shifted upward. In the robust optimum, protecting the common default at low states requires excluding nearby upward actions. The gap thereby acts as a hurdle that determines when the agent can leave the default and access the higher discretionary region.

Our paper belongs to the growing literature on robust delegation. \citet{frankel2014aligned}, \citet{guo2023regret}, \citet{alonso2026robust}, and \cite{dworczak2026robust} explore how ambiguity on the agent's preferences shapes optimal delegation. We instead fix a known upward preference bias and show how uncertainty about the distribution of the payoff-relevant state can generate a simple nonconvex delegation structure. Closely related work studies uncertainty about the state distribution while retaining its mean and obtains interval delegation, or randomization over intervals, as optimal \citep{hu2023robust,jia2026random}. The contrast shows that restrictions on the state distribution can qualitatively reshape the optimal delegation rule, underscoring the importance of how prior uncertainty is specified.

Finally, our analysis connects to the broader literature on robust mechanism design, which studies incentive design when the designer has limited knowledge of the environment. Within this literature, our approach is closely related to \cite{bergemann2008pricing} and \cite{guo2023regret,guo2025robust}, who study design without priors under a minimax-regret criterion in pricing, project choice, and monopoly regulation, respectively.\footnote{A growing body of work explores worst-case design under uncertainty about demand, available actions, information structures, or equilibrium behavior; see, for example, \cite{bergemann2011robust,carroll2015robustness,carroll2019robustness,kambhampati2023randomization, brooks2024structure, brooks2025simplicity, he2024interim,li2025robust, durandard2025robust,zhang2026generalized}.} Our analysis enriches this literature by studying a delegation environment without transfers, where incentive compatibility creates an endogenous cost of responsiveness: making the decision more responsive to the agent's information requires additional action dispersion. Solving the unrestricted robust mechanism problem reveals how this tradeoff shapes the optimal stochastic allocation and provides the benchmark that guides our characterization of optimal delegation.

\paragraph{Organization.} The rest of the paper is structured as follows. Section \ref{sec:model} sets up the model. Section \ref{sec:deterministic} discusses deterministic delegation and its limitation. Section \ref{sec:optimalmechanism} characterizes optimal robust delegation using a mechanism design approach. Section \ref{sec:extensions} considers two extensions: uncertainty about the agent's bias and a known state-dependent bias. Section \ref{sec:conclusion} concludes. Proofs of the main results are relegated to the Appendix. The Online Appendix collects additional proofs and discussion. 

\section{Model}
\label{sec:model}

A principal must make a decision using information privately held by an agent. Monetary transfers are unavailable. The principal designs a delegation rule that determines the agent's discretion over the final decision.

\paragraph {Environment.} The payoff-relevant state is \(\theta\in\Theta=[0,1]\), privately observed by the agent. The principal chooses an action \(a\in\cA=\mathbb{R}\).\footnote{Restricting the action space to the principal's set of ideal actions, $[0,1]$, can entail a loss, as in, for example, \cite{kovac2009stochastic}.} We adopt the canonical quadratic-loss specification for the payoffs of the principal and the agent:
$$v(a,\theta)=-(a-\theta)^2    \qquad \mbox{and} \qquad u(a,\theta)=-(a-\theta-b)^2.
$$
We focus on the case where $b\in(0,1)$, so the two parties' preferences are partially aligned.

\paragraph{Delegation.} 
A \textit{delegation} rule (randomly) assigns an admissible set of actions from which the agent may choose. Formally, let $\cD$ denote the collection of nonempty compact subsets of $\cA$. A (random) delegation rule is a probability distribution $\Psi\in\Delta(\cD)$, and, upon observing the realized delegation set $D$ and the state $\theta$, the agent chooses the action $a_D(\theta)$ closest to his ideal point $\theta+b$:
\[ a_D(\theta) = \min\operatorname*{arg\,min}_{a\in D}(a-\theta - b)^2.
\]
The agent breaks ties in the principal’s favor by choosing his lowest optimal action.\footnote{Our main result holds under adverse tie-breaking. See Online Appendix \ref{subsec:tie}.}  A delegation rule is \textit{deterministic} if $\Psi$ is degenerate and places probability one on a single set $D\in\cD$.

\paragraph{Robust Design.} 
The principal knows that the state lies in $\Theta$ but has no further information about its distribution. Thus, her ambiguity set contains all Borel probability distributions on $\Theta$:
\[\cF=\Delta(\Theta).\] 
The principal evaluates a delegation rule by its worst-case expected quadratic loss over distributions in $\cF$. At state $\theta$, choosing action $a$ generates the loss
\[
r(a,\theta)=(a-\theta)^2.
\]
For a given delegation set, loss can arise in two ways. If it is too restrictive, it may rule out actions that the agent would willingly choose, and that would move the decision closer to the principal's ideal action. If it is too permissive, it may give the biased agent enough latitude to choose actions against the principal's interest. Designing delegation therefore requires balancing the loss from restricting useful adaptation against the loss from granting excessive discretion.

The two sources of loss described above have a natural interpretation as forgone opportunities relative to what the principal could achieve by observing the state and choosing her ideal action. Since this complete-information benchmark entails zero loss in every state, the loss therefore measures \emph{regret} in the sense of \citet{Savage1951} and \citet{Milnor1954}. Following the regret terminology used in robust design, we henceforth refer to this loss as regret and frame the principal's robust design problem as one of \emph{minimax regret}.\footnote{Since this loss is also the principal's disutility, minimizing worst-case expected regret coincides with maximizing worst-case expected utility. We adopt the minimax-regret framing for expositional convenience.} Formally,
\[
V_D
=
\inf_{\Psi\in\Delta(\cD)}
\sup_{F\in\cF}
\int_{\Theta}\int_{\cD}
\bigl(a_D(\theta)-\theta\bigr)^2
\,d\Psi(D)\,dF(\theta).
\]
Because $\cF$ contains all degenerate distributions, the inner worst case is equivalently the worst state:
\[
\sup_{F\in\cF}
\int_{\Theta}\int_{\cD}\bigl(a_D(\theta)-\theta\bigr)^2
\,d\Psi(D)\,dF(\theta)
=
\sup_{\theta\in\Theta}\int_{\cD}\bigl(a_D(\theta)-\theta\bigr)^2\,d\Psi(D).
\]
The principal thus chooses the random delegation rule that minimizes the largest statewise regret.

\section{Deterministic Delegation and Its Limitation}
\label{sec:deterministic}
We begin with deterministic delegation, which provides a benchmark for understanding robust delegation under distributional uncertainty.
\begin{prop}
\label{prop:detvalue}
Under deterministic delegation, the principal's minimax regret is
$\min\left\{b^2,1/4\right\}$.
In particular:
\begin{itemize}
\item if $b\leq 1/2$, full delegation $D=[0,1]$ is optimal;
\item if $b>1/2$, no delegation, implemented by $D=\{1/2\}$, is optimal.
\end{itemize}
\end{prop}
Under distributional uncertainty, deterministic delegation leaves little room for a meaningful compromise between responsiveness and control. Any deterministic rule that remains responsive to the agent's information must expose the principal to the agent's full bias somewhere, and thus cannot improve upon the worst-case regret of \(b^2\) under full delegation \(D=[0,1]\). The alternative is to dispense with responsiveness altogether and fix the action in advance. The minimax constant action is \(1/2\), yielding worst-case regret \(1/4\). Thus, either full delegation or no delegation minimizes worst-case regret among deterministic delegation rules.

The following example shows how randomization can balance losses across states and strictly improve the robust guarantee.

\begin{exam*}
Suppose that \(b=1/2\). In this case, the worst-case regret of the optimal deterministic delegation is \(1/4\). Consider the random delegation rule 
\[
\Psi=
\begin{cases}
D_1=\{1/4\}, & \text{with probability } 1/3,\\
D_2=[1/2,1], & \text{with probability } 2/3.
\end{cases}
\]
The agent's choice under each set is
\[
a_{D_1}(\theta)=\tfrac14,\quad\text{and}\quad
a_{D_2}(\theta)=\min\left\{\theta+\tfrac12,1\right\}.
\]
The two delegation sets have opposite strengths across states: \(D_1\) limits the agent's upward distortion at low states, but performs poorly as the state becomes high. \(D_2\) has the opposite profile: it exposes the principal to the agent's bias at low states, but preserves the flexibility needed when the state is high. 

\begin{figure}[t]
    \centering
    \includegraphics[width=0.7\textwidth]{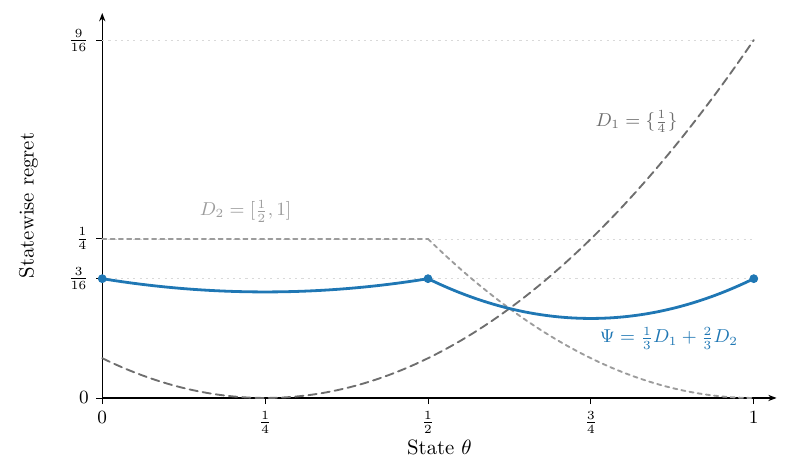}
    \vspace{-0.4cm}
    \caption{Regret from delegation sets $D_1$, $D_2$, and random delegation rule $\Psi$.}
    \label{fig:example regret}
\end{figure}

Randomization hedges these opposing statewise exposures. Figure~\ref{fig:example regret} shows how it lowers worst-case regret by combining two sets whose losses peak at different states. Under $D_1$, fixing the action at $1/4$ keeps regret at or below $1/16$ on $[0,1/2]$, but regret reaches $9/16$ at $\theta=1$. Under $D_2$, the agent chooses his ideal action on $[0,1/2]$, where regret reaches its maximum of $1/4$. At higher states, he chooses the cap of $1$, and regret falls to zero at $\theta=1$. Assigning probability $1/3$ to $D_1$ and $2/3$ to $D_2$ balances these losses across states. The mixture's worst-case regret is $3/16$, attained at $\theta\in\{0,1/2,1\}$, compared with $1/4$ under the best deterministic rule.
\end{exam*}

\section{Optimal Random Delegation}
\label{sec:optimalmechanism}

We now consider the general class of random delegation rules and characterize the optimal one.
\begin{thm}
\label{thm:optimal-delegation}
There exists an optimal random delegation rule under which each realized delegation set consists of a common default action $a_0$, possibly together with a discretionary interval $[\underline a, \bar a]$ with random boundaries, i.e., each realized delegation set takes the form 
\[D=\{a_0\}\cup[\underline a,\bar a] \ \text{ or } \ D=\{a_0\}\quad \Psi^*\text{-a.s.}\] 
The common default action $a_0=b(1-\frac12e^{1-1/b})$, and whenever the discretionary interval is present, its floor satisfies
$\underline a\ge 2b-a_0$. Moreover, the corresponding minimax regret is
\[
V_D=a_0^2<\min\{b^2,1/4\},
\]
which is the optimum of the unrestricted stochastic mechanism design problem.
\end{thm}

The optimal rule has three components: a common default action, a gap above the default, and a discretionary interval with random boundaries. The default $a_0$ is available under every realization. Whenever additional discretion is granted, the lowest additional available action is at least $2b-a_0$, leaving a fixed gap above the default. Beyond this hurdle, the agent receives discretion over an interval $[\underline a,\bar a]$, whose floor and ceiling vary across realizations. 

The random discretionary interval governs responsiveness at higher states. Its floor determines how easily the agent can leave the default, while its ceiling determines how far the action can adjust once discretion is exercised. Randomizing these boundaries generalizes the logic of the example in the previous section: different realizations place different weights on control and responsiveness, allowing the principal to hedge her statewise exposure.

The default and the gap govern the decision at low states. The default gives the principal control there; at $\theta=0$, the agent selects $a_0$ instead of his ideal action $b$. Since $a_0<b$, however, the default lies below the agent's ideal action even at the lowest state. If nearby higher actions were available, the agent would immediately move away from it. The gap makes the default effective by ruling out such modest upward deviations. In particular, $2b-a_0$ is exactly the action that leaves the agent indifferent to $a_0$ at $\theta=0$, so excluding $(a_0,2b-a_0)$ creates the hurdle required to sustain the default at low states. These two features are in fact necessary for optimality.

\begin{prop}
\label{prop:necessary-default-gap}
Every optimal random delegation rule has the same common default $a_0=b(1-\frac12e^{1-1/b})$ and excludes all actions in
$
(a_0,2b-a_0)$ 
almost surely.
\end{prop}

As the bias $b$ increases, both the common default $a_0(b)$ and the width of the required gap,
$2\bigl(b-a_0(b)\bigr)=be^{1-1/b}$, strictly increase. The higher default accommodates part of the agent's stronger upward preference, but it rises less than one-for-one with $b$. The lowest type is therefore farther from his ideal action, making nearby higher actions more attractive; a wider gap is needed to keep him at the default.

Both the default-gap structure and random discretion are necessary for robust optimality, but the randomization need not be resolved all at once. 
The optimal rule admits a sequential implementation that separates the decision to seek discretion from its subsequent approval. The principal first announces a randomly drawn floor for the potential discretionary interval. After observing this floor, the agent chooses whether to retain the default or proceed. Proceeding does not guarantee access to discretion: a subsequent approval lottery determines whether the interval becomes available or the default is implemented.

\begin{prop}%[Sequential implementation]
\label{prop:sequential-implementation}
The random delegation rule $\Psi^*$ has an outcome-equivalent sequential implementation as follows. Let $q$ be the probability that the discretionary interval is present under $\Psi^*$, and draw a potential interval $(\underline a,\bar a)$ from its joint distribution conditional on being present.
\begin{enumerate}[label=(\roman*)]
\item \emph{Stage 1.} The floor $\underline a$ is announced. The agent may retain $a_0$ or continue. 
\item \emph{Stage 2.} Conditional on continuation, $a_0$ is implemented with probability $1-q$. With probability $q$, discretion is approved and the agent chooses from $[\underline a,\bar a]$.
\end{enumerate}
\end{prop}

The implementation separates the agent's decision whether to retain the default from the subsequent choice conditional on approval. For this implementation, it is essential that the agent observes the realized floor of the discretionary interval before deciding whether to seek approval, as it preserves the agent's realization-by-realization choice under the original rule; without this information, the agent would instead compare the default with the average continuation lottery.\footnote{A procedure that asks the agent whether to continue without revealing the realized floor is also well defined, but it is generally not outcome-equivalent to $\Psi^*$. Disclosure allows the agent to continue only for floor realizations under which discretion is preferable to the default. Without disclosure, he must accept or reject the continuation lottery as a whole. Thus, if $\Delta$ denotes the gain from discretion in a given realization, disclosure yields the option value $\mathbb E[\max\{0,\Delta\}]$, whereas nondisclosure yields $\max\{0,\mathbb E[\Delta]\}$. The former is weakly larger, and strictly larger when $\Delta$ changes sign across floor realizations. Nondisclosure therefore generally cannot reproduce the state-contingent action distribution induced by $\Psi^*$ or its robust value.} For every realized $[\underline a,\bar a]$, the interval is preferable to the default exactly when the agent's ideal action lies above the midpoint between $a_0$ and the floor (i.e., $\theta+b>(a_0+\underline a)/2$). The approval probability scales the expected gain from seeking discretion but does not change its sign.

This alternative timing remains a delegation mechanism: whenever discretion is granted, the agent freely chooses from an interval generated by the principal's committed rule. It preserves a key practical advantage of delegation by requiring little communication. The principal need not elicit a detailed report of the agent's private information; after observing the hurdle, the agent only decides whether to leave the default and, if discretion is granted, chooses freely within the resulting interval. This gives the optimal mechanism a natural approval-stage interpretation. The principal first announces the floor, and the agent then decides whether to advance the project. A project that advances may still be returned to the default, while an approved project receives a budget range within which the agent chooses its scale. Appendix~\ref{app:sequential-implementation} establishes the state-by-state equivalence.

Applied to firm investment, our optimal delegation rule offers a simple explanation for lumpy investment. Specifically, capital expenditure tends to cluster at very low and relatively high levels, with comparatively little activity at intermediate scales \citep{doms1998capital,cooper2006nature}. The quantitative macro investment literature commonly captures such behavior through nonconvex adjustment costs.

Our mechanism provides a principal--agent microfoundation for such frictions: robust incentive considerations in delegated capital allocation can endogenously make intermediate investment scales unavailable. This mechanism has a direct institutional counterpart in firms' internal capital-budgeting processes. Capital-allocation authority is often delegated to managers with superior local information \citep{graham2015capital}, while firms commonly use hurdle rates in project evaluation \citep{graham2001theory}. 
Practitioner accounts describe stage-gate procedures in which projects must demonstrate sufficient value before substantial capital is committed \citep{barshop2016capital}. Projects that fail to clear a gate may be held, canceled, or returned for revision, while those that proceed face much larger capital commitments at later stages. The resulting choice structure mirrors our mechanism: the status quo is maintained until a project clears a hurdle, after which substantially greater capital commitments become feasible. Importantly, the same practitioner account links these controls to agency concerns. Project sponsors may downplay risks out of self-interest, in part because larger capital commitments allow them to build larger businesses and may improve compensation and promotion prospects.

The remainder of this section derives the optimal random delegation rule in Theorem~\ref{thm:optimal-delegation}. The challenge is that the principal chooses a distribution over feasible-action sets, and the agent's behavior depends on the entire structure of each realized set. We therefore recast the problem in mechanism design terms, first characterizing the optimal stochastic incentive-compatible outcome rule and then showing that it admits an implementation through random delegation.

\subsection{Stochastic Mechanisms}
\label{subsec:stochastic-mechanisms}

Now we allow the principal to choose an arbitrary stochastic mechanism \citep{kovac2009stochastic,kleiner2021extreme,kolotilin2025persuasion}. The standard revelation principle argument applies, so it is without loss to focus on stochastic direct mechanisms. The principal may choose any randomized direct mechanism, represented
by a Borel probability kernel
\[
    M:\Theta\longrightarrow\Delta(\mathcal{A}).
\]
A report $\hat\theta$ induces a lottery
$M_{\hat\theta}:=M(\hat\theta)$ over actions. We require each lottery to have a finite second moment. We say a mechanism is incentive compatible if 
\begin{equation*}
\int_{\R}(a-\theta-b)^2\,\dd M_\theta(a)
\leq
\int_{\R}(a-\theta-b)^2\,\dd M_{\hat\theta}(a)
\qquad\text{for all }\theta,\hat\theta\in\Theta.
\end{equation*}
Let $\cM$ denote the class of these mechanisms. 
For a mechanism $M$ and a state distribution $F$, define the principal's regret by
\[
R(M,F)
=
\int_\Theta\int_{\cA}(a-\theta)^2\,\dd M_\theta(a)\,\dd F(\theta).
\]
When $F$ is degenerate at $\theta$, we write
$R(M,\theta)=\int_{\cA}(a-\theta)^2\,\dd M_\theta(a)$
for the corresponding statewise regret.

The unrestricted robust design problem is therefore
\begin{equation*}
V^*=\inf_{M\in\cM}\sup_{F\in\cF}R(M,F)
=\inf_{M\in\cM}\sup_{\theta\in\Theta}R(M,\theta).
\end{equation*}
The second equality is the mechanism design counterpart of the statewise representation in Section~\ref{sec:model}. We verify optimality by constructing a Nash equilibrium of a zero-sum game between the principal and adversarial nature. For any distribution \(F\), the problem of minimizing \(R(M,F)\) over incentive-compatible mechanisms provides a lower bound on the minimax value, while any mechanism \(M\) provides an upper bound through its worst-case regret. We construct \(M^*\) and \(F^*\) for which these bounds coincide, thereby establishing a saddle point such that
\[
R(M^*,F)\le R(M^*,F^*)\le R(M,F^*)
\]for every \(F\in\cF\) and every incentive-compatible mechanism \(M\in\cM\).

With quadratic losses, each lottery over actions matters only through its mean and variance \citep{goltsman2009mediation,kovac2009stochastic}, making the problem tractable. Define the conditional mean and variance of the action by
\[
\mu(\theta)=\int_{\R}a\,\dd M_\theta(a),
\qquad
\tau(\theta)=\int_{\R}\bigl(a-\mu(\theta)\bigr)^2\,\dd M_\theta(a).
\]
When the state is $\theta$ and the agent reports truthfully, the
principal's regret is
\begin{equation*}
R(M,\theta)=\bigl(\mu(\theta)-\theta\bigr)^2+\tau(\theta).
\end{equation*}
An agent with true state $\theta$ who reports $\hat\theta$
receives expected utility
\begin{equation*}
U(\hat\theta\mid\theta)
=-\bigl(\mu(\hat\theta)-\theta-b\bigr)^2
-\tau(\hat\theta).
\end{equation*}
The first term in the principal's loss is the distortion in the average action; the second is the cost of dispersion around that average. Variance is costly to both players and therefore acts like money burning: it destroys surplus, but can improve incentives by making reports that induce more attractive average actions less appealing to the agent.\footnote{\cite{ambrus2017delegation} and \cite{amador2020money} discuss the role of money burning in delegation.} Because no higher moments enter either player's payoff, we can characterize the robust optimum through \((\mu,\tau)\).

The following mean--variance characterization of incentive compatibility adapts
a standard result from \citet{goltsman2009mediation} and \citet{kovac2009stochastic} to our
notation.
\begin{lem}[Incentive compatibility]
\label{lem:ICmoments}
There exists an incentive-compatible stochastic mechanism with conditional mean $\mu$ and variance $\tau$ if and only if
\begin{itemize}
\item $\mu$ is nondecreasing; 
\item the agent's truthful utility $U(\theta):=U(\theta\mid\theta)$ satisfies
\begin{equation}
\label{eq:IC-envelope}
U(\theta)=U(0)+2\int_0^\theta\bigl(\mu(x)-x-b\bigr)\,\dd x;
\end{equation}
\item $\tau(\theta) \geq 0$ for every $\theta \in [0,1]$.
\end{itemize}
\end{lem}

Monotonicity reflects the agent's preference for higher actions at higher states. The envelope condition determines the dispersion required to sustain a given response of the mean. Wherever the moment schedules are differentiable,
\begin{equation}\label{eq:var_derivative}
    \tau'(\theta)
=
2\bigl(\theta+b-\mu(\theta)\bigr)\mu'(\theta).
\end{equation}
As long as the mean remains below the agent's ideal action, responsiveness to the agent's information, \(\mu'(\theta)>0\), requires strictly increasing variance. Keeping the mean locally flat, by contrast, avoids this additional variance cost. Since matching the agent's ideal action in any state would already deliver the full-delegation regret \(b^2\), any mechanism that improves on that benchmark operates in this region.

The mean--variance reformulation thus recasts the mechanism design problem as a tradeoff between responsiveness and dispersion. The optimal mechanism lets the mean rise while the value of adaptation justifies the associated increase in variance, and eventually flattens it when further responsiveness becomes too costly.
\subsection{Optimal Mechanism}
\label{subsec:optimal-robust-mechanism}

Using the mean–variance representation, we characterize the robust optimum by constructing a saddle point \((M^*,F^*)\), where $M^*=(\mu^*,\tau^*)$.

\begin{prop}[Optimal mechanism]
\label{prop:optimal mech}
\label{thm:optimalmechanism}
\label{prop:uniqueness}
Optimal mechanisms exist, and all optimal mechanisms induce the same conditional mean and variance $(\mu^*,\tau^*)$. These moments have the following properties:
\begin{enumerate}
\item At $\theta=0$, the action is deterministic, with $\mu^*(0)=a_0=b(1-\frac12e^{1-1/b})$ and $\tau^*(0)=0$.
\item On $(0,1-b)$, both $\mu^*$ and $\tau^*$ are strictly increasing, while the upward distortion $\mu^*(\theta)-\theta$ is strictly decreasing. Regret is constant throughout the interval.
\item On $[1-b,1]$, both $\mu^*$ and $\tau^*$ are constant, with $\mu^*(\theta)=1-b/2$. Regret is highest at the two endpoints and strictly lower in the interior.
\end{enumerate}
The unrestricted robust value is
$
V^*=a_0^2<\min\{b^2,1/4\}.$
\end{prop}

The proposition identifies the mean and variance shared by all optimal mechanisms. Different action lotteries can deliver these moments and thus generate the same payoffs and incentives. It suffices to implement the optimal mean and variance without reproducing any particular family of action lotteries. Appendix~\ref{app:mechanism-details} gives the closed-form expressions for the optimal moments. Figure~\ref{fig:optimal-moments} illustrates the construction.

\begin{figure}[t]
\begin{center}
\includegraphics[width=\textwidth]{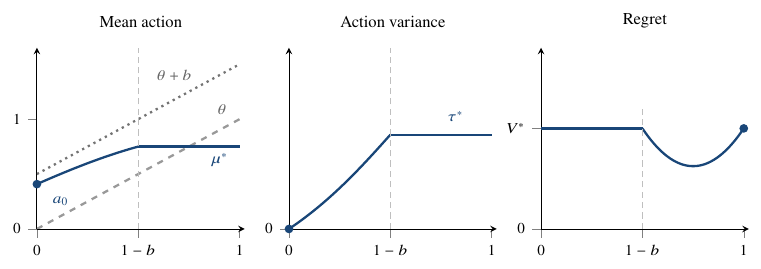}
\caption{Optimal mean, variance, and regret across states.}\label{fig:optimal-moments}
\medskip
\end{center}
\scriptsize{Notes: The figure shows the case $b=1/2$. The mean and variance increase up to $1-b$ and are constant thereafter. Regret equals $V^*$ on $[0,1-b]$ and at $\theta=1$, and is strictly lower on $(1-b,1)$.
}
\end{figure}

The agent at state $1-b$ and the principal at state $\theta=1$ have the same ideal action and therefore rank lotteries in the same order. Incentive compatibility requires type $1-b$ to weakly prefer his assigned lottery to any lottery assigned to a higher report. The principal at state $\theta=1$ shares this preference. It is therefore optimal to pool all higher states at a common mean $\mu$, since additional responsiveness above $1-b$ would serve only types whose ideal actions exceed the highest action the principal could want. Thus, there exists an optimal mechanism that pools all states above some cutoff $\kappa\leq 1-b$. Since $V^*<b^2$, we also have $\kappa>1-2b$.

By the envelope condition, a constant mean $\mu$ on $[\kappa,1]$ implies a constant variance $\tau$ on that interval. The variance at state $0$ determines its level, $\tau(\kappa)$. Since all pooled states share the same mean and variance, regret on $[\kappa,1]$ varies only with the distance between the state and the pooled mean. It is therefore largest at one of the two endpoints. At the optimum, the pooled mean equalizes regret at the two endpoints to minimize the maximum. The pooled mean is therefore the midpoint, $\frac{1+\kappa}{2}$. This gives the boundary condition for constructing the responsive region.

To address how the mean action changes across states, we differentiate the regret decomposition and insert the differentiated envelope condition, which yields
\begin{equation}
\label{eq:regret-derivative}
R'(M,\theta)=2b\mu'(\theta)-2(\mu(\theta)-\theta).
\end{equation}
The first term captures the cost of responsiveness. Raising the mean increases both the upward distortion and, through incentive compatibility, the variance. The second term reflects the level of upward distortion. A larger distortion reduces the derivative of regret.

For a given cutoff $\kappa$, the mean must rise from its initial value at $\theta=0$ to the pooled mean at $\kappa$. A lower initial mean reduces distortion at state $\theta=0$, but requires a larger cumulative rise in $\mu$ on $[0,\kappa]$. If regret fell below its maximum at some state, the mean could rise faster there and start from a lower value at $\theta=0$, thereby lowering regret. Regret therefore remains at its maximum throughout the responsive region, and
\[ \mu'(\theta)= \frac{\mu(\theta)-\theta}{b}. \]
The responsiveness of the mean action is therefore proportional to the remaining upward distortion. As the distortion falls, the mean rises more slowly.
Solving the differential equation backward from the boundary condition $\mu_\kappa(\kappa)=\frac{1+\kappa}{2}$ gives
\begin{equation*} 
\mu_\kappa(\theta)=\theta+b-\Bigl(b-\frac{1-\kappa}{2}\Bigr)
\exp\left(-\frac{\kappa-\theta}{b}\right) \quad \forall \theta \in [0,\kappa].
\end{equation*}
As the state falls below the cutoff, the distance between the agent's ideal action and the implemented mean shrinks exponentially at rate $1/b$. At the lowest state, the mean is
\[ a_0(\kappa)=\mu_\kappa(0)=b-\Bigl(b-\frac{1-\kappa}{2}\Bigr)e^{-\kappa/b}. \]
Given an initial variance $\tau(0)$, regret under cutoff $\kappa$ equals $a_0(\kappa)^2+\tau(0)$, so the optimal cutoff minimizes $a_0(\kappa)$. Differentiating $a_0(\kappa)$ yields a minimum at $\kappa=1-b$. 
Therefore, the optimal mean equals
\[ \mu^*(\theta)= \theta+b- \frac{b}{2}\exp\left( -\frac{1-b-\theta}{b} \right) \quad \forall \theta \in[0,1-b],\]
and at $\theta=0$, the value of the default action $a_0$ equals
$ a_0=b \left( 1- \frac{\exp(1-1/b)}{2}\right).$

By the envelope condition, incentive compatibility determines how variance changes across reports, leaving its level free. Subtracting the same constant from the variance at every state preserves incentives and lowers regret. Since variance must be nonnegative, its minimum must equal zero at the optimum. The mean $\mu$ must lie below the agent's ideal action $\theta+b$, so an increase in the mean across reports requires an increase in variance to deter the agent from overstating the state. Variance is therefore smallest at the lowest state, and $\tau^*(0)=0$. The lowest type thus receives a deterministic action, which is the common default under delegation.

This argument mirrors a familiar one in mechanism design with transfers. For a given allocation rule, the envelope condition determines transfers up to a constant. A binding constraint at an extreme type determines that constant. Here, variance acts as money burning and plays the role of transfers. The nonnegativity constraint on variance binds at the lowest type.

Since the action at the lowest state is deterministic, its regret comes purely from the distortion and equals $a_0^2$. As regret is constant on $[0,1-b]$, this pins down the robust guarantee $V^*=a_0^2$. 

The mechanism in Proposition~\ref{prop:optimal mech} forms a saddle point with the distribution 
\[ F^*=\frac12 e^{1-1/b}\delta_0
+\frac{1-\frac12e^{1-1/b}}{b} e^{-\theta/b}\mathbf{1}_{(0,1-b)}(\theta)\,\dd\theta
+e^{1-1/b}\left(1-\frac12e^{1-1/b}\right)\delta_1.\]
Under $F^*$, $M^*$ minimizes expected regret among all incentive-compatible mechanisms.
The support of $F^*$ is $[0,1-b]\cup\{1\}$, precisely the set of states where regret under $M^*$ is highest.
The distribution has atoms at $0$ and $1$ and a strictly positive density on $(0,1-b)$.
Throughout the responsive region, the conditional mean of the state above each cutoff equals the cutoff type's ideal action,
\[
\mathrm{E}_{F^*}[x \mid x > \theta] = \theta + b,
\qquad 0 \leq \theta \leq 1-b.
\]
Under $F^*$, the expected state conditional on exceeding $\theta$ equals the agent's ideal action, $\theta+b$. This implies that the hazard rate is constant at $1/b$ on $(0,1-b)$. In the terminology of \citet{alonso2008optimal}, the forward bias is zero throughout this interval. Intuitively, a principal considering a cap on discretion at state $\theta$ would herself choose $\theta+b$ under $F^*$. In this sense, the agent's upward bias is indistinguishable from his information at each state in the responsive region. The principal is therefore indifferent to marginal changes $d\mu(\theta)$ in the mean action over $[0,1-b]$. Under $F^*$, the remaining terms yield the lower bound 
\begin{equation}
\label{eq:saddle-certificate}
\begin{split}
R(M,F^*)-V^*
={}&\bigl(\mu(0)-a_0\bigr)^2+\tau(0)+2F^*(\{1\})\int_{1-b}^1\bigl(\mu(1)-\mu(\theta)\bigr)\,\dd\theta
\geq0.
\end{split}
\end{equation}
Each term has a direct interpretation. Against $F^*$, a mechanism loses by moving its initial mean away from $a_0$, by adding variance at the bottom, or by continuing to raise the mean over the upper region. The last term is nonnegative because incentive compatibility requires a nondecreasing mean. The candidate makes all three terms zero and keeps regret at most $V^*$ in every state. Hence,
$
R(M^*,F)\leq R(M^*,F^*)=V^*\leq R(M,F^*)$ for every $F\in\cF\text{ and }M\in\cM$.

The same argument pins down the moments of every optimum. Any optimal mechanism must make each term in~\eqref{eq:saddle-certificate} zero and attain regret $V^*$ wherever $F^*$ places weight. Appendix~\ref{app:mechanism-details} provides the formal argument, which also covers non-differentiable moment schedules. 

\subsection{Implementation by Random Delegation}
\label{subsec:implementation-random-delegation}
\label{sec:implementation}
\label{subsec:optimal-delegation}

We now show how random delegation delivers the optimal mean and variance. A random delegation rule is a special stochastic mechanism: conditional on each realized set \(D\), the induced action \(a_D(\theta)\) must be the agent's optimal choice from that set at every state. Hence every random delegation rule induces an interim incentive-compatible stochastic mechanism, and therefore
$
V^*\leq V_D$. 
The class of mechanisms induced by random delegation is, however, a proper subset of the class of stochastic incentive-compatible mechanisms. Random delegation imposes
an additional consistency requirement across states: the supports of the induced action lotteries cannot be chosen independently across states, since all states face the same realized delegation set. The substantive question is whether this additional structure is costly for the principal. Under a direct mechanism, the agent compares lotteries before the action is drawn. Under delegation, he chooses an action after the set is realized, so optimality must hold realization by realization. We show that despite this stronger requirement, a common default and one discretionary interval suffice to reproduce the optimal mean and variance and attain \(V^*\).

\paragraph{Necessary restrictions.}
The optimal moments first identify what almost every realized delegation set must have in common. Since $\mu^*(0)=a_0$ and $\tau^*(0)=0$, almost every set must induce the same action $a_0$ at the lowest state. Hence $a_0$ must belong to almost every realized set. At this state, the agent's ideal action is $b$, and every action in $(a_0,2b-a_0)$ is strictly closer to $b$ than the default. All such actions must therefore be excluded. This establishes Proposition~\ref{prop:necessary-default-gap}.

These restrictions already rule out optimality of random interval delegation. An interval containing $a_0$ cannot offer actions above the gap without also filling in the gap. Any actions below $a_0$ would never be chosen, since every state's ideal action exceeds $a_0$. Thus, a rule supported on intervals and satisfying the necessary restrictions would induce $a_0$ at every state, contradicting the responsiveness of $\mu^*$ on $(0,1-b)$.

The remaining question is whether the necessary default and gap leave enough flexibility to implement the optimal moments. We establish sufficiency using sets of the form
$
D=\{a_0\}$, or
$D=\{a_0\}\cup[\underline a,\bar a]$, where $2b-a_0\leq\underline a\leq\bar a$. 
The interval may be degenerate. The argument proceeds by representing the random endpoints, matching the mean, and then showing that the variance follows automatically.

\paragraph{Representing the random endpoints.}
Let $q$ be the probability that a discretionary interval is offered. Represent the endpoints by their unconditional cumulative probabilities,
\[
G(x)=\Pr_\Psi\bigl(D\neq\{a_0\},\ \underline a\leq x\bigr),
\qquad
H(x)=\Pr_\Psi\bigl(D\neq\{a_0\},\ \bar a\leq x\bigr).
\]
Both have total mass $q$. Since $\underline a\leq\bar a$, the upper endpoint must be larger in the sense of first-order stochastic dominance:
\[
H(x)\leq G(x)\quad\text{for every }x,
\qquad
G(\infty)=H(\infty)=q.
\]
Conversely, suppose $q>0$ and normalize the two subdistributions to obtain the CDFs $G/q$ and $H/q$. Strassen's monotone coupling theorem implies that $H\leq G$ is sufficient for these marginals to admit a coupling satisfying $\underline a\leq\bar a$ almost surely \citep{strassen1965existence}. In one dimension, the coupling is explicit: draw a common uniform variable $Z\in(0,1)$ and choose the $Z$-quantiles of $G/q$ and $H/q$. With the remaining probability $1-q$, offer only $\{a_0\}$. Thus, $(q,G,H)$ provides an equivalent and more tractable representation of the distribution over delegation sets. The case $q=0$ corresponds to offering only $\{a_0\}$.

\paragraph{Induced moments.}
For a given state, the agent either stays at the default, moves to the floor, chooses his ideal action inside the interval, or is constrained by the ceiling. He leaves the default precisely when
\[
\theta+b>\frac{a_0+\underline a}{2}.
\]
These four choices express the induced moments directly in terms of the endpoint distributions. Taking $G$ to be absolutely continuous, as in our construction, for $k=1,2$ we have
\begin{equation}
\label{eq:delegation-induced-moments}
\begin{aligned}
\E_\Psi[a_D(\theta)^k]
={}&a_0^k\Bigl[1-G\bigl(2(\theta+b)-a_0\bigr)\Bigr]
+\int_{\theta+b}^{2(\theta+b)-a_0}a^k\,\dd G(a)\\
&+(\theta+b)^k\bigl[G(\theta+b)-H(\theta+b)\bigr]
+\int_{-\infty}^{\theta+b}a^k\,\dd H(a).
\end{aligned}
\end{equation}
The four terms correspond, in order, to choosing the default, the floor, the agent's ideal action, and the ceiling. In particular, the moments depend only on the two endpoint distributions, so we can choose a coupling that respects the order of the endpoints without imposing additional moment restrictions.

\paragraph{Matching the mean.}
Write $\mu$ and $\tau$ for the moments induced by the delegation rule. Differentiating the first moment gives, almost everywhere,
\begin{equation}
\label{eq:endpoint-response}
\mu'(\theta)=G(\theta+b)-H(\theta+b)
+4(\theta+b-a_0)G'\bigl(2(\theta+b)-a_0\bigr).
\end{equation}
The two terms capture two distinct ways in which the agent's information affects the decision. The difference $G(\theta+b)-H(\theta+b)$ is the probability that his ideal action lies inside the discretionary interval; in these realizations, his action tracks the state. The second term captures switches across the gap. A switch at state $\theta$ occurs at floor $\underline a=2(\theta+b)-a_0$ and raises the action by $2(\theta+b-a_0)$. As the state increases, floors cross this switching threshold at rate $2G'(2(\theta+b)-a_0)$.

Since the default and gap already ensure $\mu(0)=a_0=\mu^*(0)$, matching $\mu^{*\prime}$ is enough to match the entire mean schedule. The two endpoint distributions allow the principal to combine adjustment within intervals with switches out of the default. She need only make their combined response equal the desired response of the mean. This leaves flexibility in implementation: the mean does not separately prescribe how much responsiveness must come from each source. Appendix~\ref{app:endpoint-construction} gives a pair of valid distributions satisfying the ordering and mass restrictions above and reproducing $\mu^*$.

\paragraph{Matching the variance.}
The second moment does not impose an independent target on the endpoint distributions. Differentiating~\eqref{eq:delegation-induced-moments} for $k=2$ gives
\[
\frac{\dd}{\dd\theta}\E_\Psi[a_D(\theta)^2]
=2(\theta+b)\mu'(\theta).
\]
There is a direct incentive explanation for this identity. Within an interval, the chosen action equals $\theta+b$, so a marginal increase in the action raises its square at rate $2(\theta+b)$. At a switch from $a_0$ to $\underline a$, indifference requires $a_0+\underline a=2(\theta+b)$, and hence
$
\underline a^2-a_0^2
=2(\theta+b)(\underline a-a_0).$ 
Thus, both sources of responsiveness generate the same relationship between changes in the first and second moments.

Using $\tau=\E_\Psi[a_D^2]-\mu^2$, we obtain
$
\tau'(\theta)=2\bigl(\theta+b-\mu(\theta)\bigr)\mu'(\theta)$, 
exactly the incentive-compatibility relation satisfied by the optimal mechanism. Once the mean matches $\mu^*$, the induced variance has the same derivative as $\tau^*$. The common default and gap also give the same initial value, $\tau(0)=0=\tau^*(0)$. The variance therefore matches at every state. The role of the default and gap is consequently twofold: they discipline choices at low states and provide the boundary condition that makes matching the mean sufficient for matching both moments.

\begin{prop}[Implementation]
\label{thm:delegation-implementation}
There is a random delegation rule $\Psi^*$ supported on sets $\{a_0\}$ or $\{a_0\}\cup[\underline a,\bar a]$, with $\underline a\geq2b-a_0$, such that
\[
\E_{\Psi^*}[a_D(\theta)]=\mu^*(\theta),
\qquad
\Var_{\Psi^*}[a_D(\theta)]=\tau^*(\theta)
\qquad\text{for every }\theta\in\Theta.
\]
Consequently, $V_D=V^*$, and $(\Psi^*,F^*)$ is a saddle point for the delegation problem.
\end{prop}

The rule preserves both players' expected payoffs at every state and therefore attains the unrestricted robust guarantee. It need not reproduce the action distribution that a given optimal direct mechanism induces at each state; it suffices to match the mean and variance, which all optimal mechanisms share. Online Appendix~\ref{subsec:payoff-implementation} explains this distinction.

\section{Extensions}
\label{sec:extensions}

The baseline model combines complete uncertainty about the state distribution with a known constant upward bias. This simple preference specification keeps the focus on distributional uncertainty. We now vary what the principal knows about the bias while preserving its upward direction. These extensions assess how far the default-and-gap structure with random discretion continues to characterize optimal delegation beyond the baseline setting. 

\subsection{Uncertain Bias}
\label{sec:uncertain-bias}

We first weaken the principal's knowledge of the agent's bias. Rather than knowing its exact magnitude, she knows only an upper bound on the upward conflict. To put this extension in context, recall the capital-budgeting setting discussed above. Headquarters may have limited knowledge of both the division's market and its manager's preferences, knowing only that the manager has an empire-building motive of uncertain strength.

Formally, let $\widetilde b\in[0,b]$ denote the agent's realized upward bias, where $b\in(0,1)$. The principal has no information about the joint distribution of $\theta$ and $\widetilde b$ beyond its support. Thus her ambiguity set is
\[
\widetilde{\cF}_b=\Delta\bigl([0,1]\times[0,b]\bigr).
\]
This formulation allows the state and the bias to be arbitrarily dependent and includes a state-independent bias value in $[0,b]$ as a special case. Thus, the principal commits to a random delegation rule against the uncertainty over both the state and the magnitude of misalignment. 

For a delegation set $D$, let $a_D(t)$ denote the action selected by an agent with ideal action $t$. The robust design problem becomes
\begin{equation}
\label{eq:uncertain-bias-problem}
\inf_{\Psi}\ \sup_{\widetilde F\in\widetilde{\cF}_b}
\E_{\widetilde F,\Psi}\!\left[
\bigl(a_D(\theta+\widetilde b)-\theta\bigr)^2
\right].
\end{equation}

\begin{prop}[Uncertain bias]
\label{prop:uncertain-bias}
The optimal random delegation rule in Theorem~\ref{thm:optimal-delegation}, calibrated to the upper bound $b$, solves~\eqref{eq:uncertain-bias-problem}.
\end{prop}

Proposition~\ref{prop:uncertain-bias} gives a worst-case interpretation of the constant-bias solution. Calibrating the delegation rule to the largest possible bias protects the principal uniformly across all states and all smaller upward biases.

\subsection{State-Dependent Bias}
\label{subsec:state-dependent-bias}
\label{subsec:affine-bias}

We now allow the degree of preference misalignment to vary with the state. The principal knows how the bias depends on the state but remains uncertain about the state. In the capital-budgeting setting, headquarters may understand how the manager's empire-building motive varies with market conditions while remaining ignorant of those conditions. 
 
We maintain the quadratic-loss specification but allow the agent's state-contingent ideal action to take a more general form, i.e., \[
u(a,\theta)=-\bigl(a-y_A(\theta)\bigr)^2,
\]
where $y_A$ is continuously differentiable and strictly increasing. To maintain that the agent is biased upward, we assume that the bias $b(\theta)=y_A(\theta)-\theta$ is nonnegative and weakly monotone, with at most one state in which it is zero. We further assume $y_A(0)<1$ to make the case interesting.\footnote{If $y_A(0)\geq1$, the deterministic action $1/2$ is robustly optimal.} 

A leading example is the affine bias specification studied by \cite{melumad1991communication}. In this case, the agent's ideal action is
\[
y_A(\theta)=(1+\beta)\theta+b_0.
\]
The resulting bias function, $b(\theta)=b_0+\beta\theta$, captures how preference misalignment varies with the state through two parameters. Here $b_0$ measures disagreement at the lowest state, while $\beta$ determines how disagreement changes with the state. The agent's ideal action rises faster than the principal's when $\beta>0$ and slower when $\beta<0$. The following proposition characterizes optimal delegation under affine bias, showing when and to what extent the baseline structure is preserved.

\begin{prop}[Optimal delegation with affine bias]
\label{prop:affine-bias}
Suppose $b(\theta)=b_0+\beta\theta$, where $0\leq b_0<1$ and $b_0+\beta>0$. Random delegation attains $V_D=V^*$, the optimal value of the unrestricted stochastic mechanism design problem. There exists a cutoff $b_c(\beta)$ such that an optimal random delegation rule takes one of the following forms:
\begin{enumerate}[label=(\roman*)]
\item \emph{Default, gap, and random interval.}
If $b_0>b_c(\beta)$, the rule offers a fixed default $a_0<b_0$, possibly supplemented by a random interval $[\underline a,\bar a]$ satisfying
$2b_0-a_0\leq\underline a$ and $\bar a\leq2-a_0$.

\item \emph{Random cap.}
If $b_0\leq b_c(\beta)$, the rule offers $D=[b_0,\bar a]$ where $\bar a$ is a random cap. 
\end{enumerate}
\end{prop}

\begin{figure}[t]
    \centering
    \includegraphics[width=0.9\linewidth]{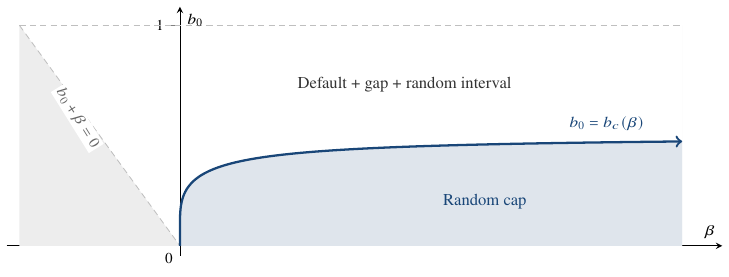}
    \caption{Optimal delegation with affine bias.}
    \label{fig:general-affine-regions}
    \medskip

    {\scriptsize
    Notes: The figure shows the affine bias $b(\theta)=b_0+\beta\theta$,
    with $0\leq b_0<1$ and $b_0+\beta>0$. The threshold $b_c(\beta)$
    is given in Appendix~\ref{app:affine-bias}.
    \par}
\end{figure}

Figure~\ref{fig:general-affine-regions} shows how the optimal delegation rule depends on disagreement at the lowest state and its change with the state. With constant or decreasing bias, the baseline structure is preserved: a fixed default is separated by a gap from a discretionary interval with random boundaries. The same structure remains optimal with increasing bias when disagreement at the lowest state is large enough ($b_0>b_c(\beta)$). When $b_0\leq b_c(\beta)$, a random cap is optimal. Here, the rule offers a connected interval with a fixed lower endpoint $b_0$ and a random upper endpoint. There is no separate default below $b_0$ and no gap separating it from higher discretionary actions.

In both cases, the principal guarantees a set of default actions while randomizing how much additional discretion to grant. The key difference lies in the default set, which can be a single action or an interval that grants the agent limited discretion. Thus, the two distinct design margins in the baseline model reflect a more general distinction between guaranteed and additional discretion.

To understand the form of optimal delegation in the two cases, we again examine the corresponding optimal stochastic mechanism and its implementation. Incentive compatibility still requires the mean action $\mu$ to be nondecreasing. As in the baseline model, whenever $\mu$ rises while remaining below the agent's ideal action, the action variance $\tau$ must also rise, increasing the principal's loss. Giving low types their ideal actions without randomization avoids this variance cost, but exposes the principal to a bias loss of $b(\theta)^2$. The two cases differ in whether the principal can afford these bias losses at low states.

Figure~\ref{fig:general-bias-means} illustrates the resulting difference. In the left panel, $y_A(\theta)=\tfrac34\theta+\tfrac12$. The bias losses at low states are too large to accept, so the mechanism assigns the lowest type a deterministic action $a_0<b_0$, as in the baseline model. In the right panel, $y_A(\theta)=\tfrac32\theta+\tfrac1{10}$, and the bias losses at low states are small enough to accept. The mechanism therefore lets low types take their ideal actions, keeping $\tau=0$ even as $\mu$ rises. This postpones the accumulation of variance until higher states. Let $\underline\theta$ denote the upper endpoint of this accommodation region, with $\underline\theta=0$ when the region is empty. Both mechanisms eventually pool states above a cutoff $\bar\theta$.

\begin{figure}[t]
    \centering
    \includegraphics[width=\linewidth]{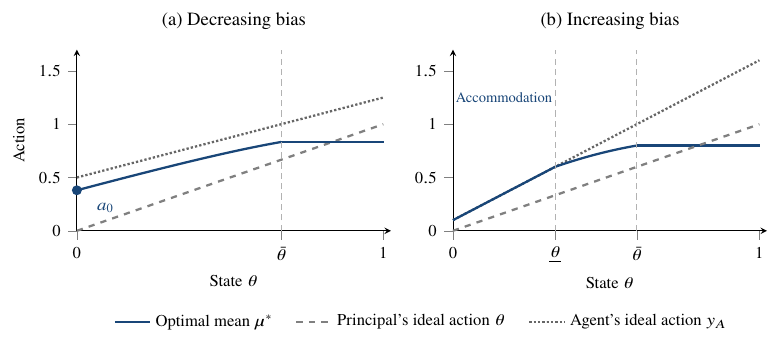}
    \caption{Optimal mean actions under decreasing and increasing affine bias.}
    \label{fig:general-bias-means}
\end{figure}

These low-state allocation outcomes explain the corresponding delegation rules. In the left panel, implementing the deterministic action $a_0<b_0$ requires every realized set to contain $a_0$ and exclude $(a_0,2b_0-a_0)$, as in the baseline model. In the right panel, accommodating low types requires every realized set to contain their ideal actions, yielding the common interval $[b_0,y_A(\underline\theta)]$. There is no lower default to protect, so no gap is needed. 

The next proposition extends these structural conclusions to general nonnegative monotone biases.

\begin{prop}[Optimal random delegation]
\label{prop:general-state-dependent-bias}
There exists an optimal random delegation rule supported on compact subsets of \([a_0,2-a_0]\), for some action \(a_0\in[0,y_A(0)]\), whose realized sets satisfy the following properties:
\begin{enumerate}[label=(\roman*)]
\item If \(a_0<y_A(0)\), every set contains the default \(a_0\) and excludes \((a_0,2y_A(0)-a_0)\).

\item If \(a_0=y_A(0)\), there exists a cutoff \(\underline\theta\in[0,1]\) such that every set contains the interval \([a_0,y_A(\underline\theta)]\).
\end{enumerate}
This rule attains the optimal value of the unrestricted stochastic mechanism design problem, so \(V_D=V^*\).
\end{prop}

For a general bias, the distinction between the two cases is less readily expressed as a condition on the bias function. Economically, it still turns on whether the principal finds it optimal to let low types choose their ideal actions. If the resulting bias losses are too costly, the lowest type is held below his ideal action, requiring a default protected by a gap. If these losses can be tolerated, low types can receive their ideal actions, allowing the common default set itself to provide limited discretion.

\section{Conclusion}
\label{sec:conclusion}

This paper characterizes robust delegation to an upward-biased agent when the principal knows the support of the state but has no prior over its distribution. Deterministic delegation essentially reduces to a choice between full discretion and a fixed decision, whereas random delegation strictly improves upon deterministic rules by hedging against uncertainty and balancing regret across states. In the baseline model with a known, state-independent bias, the optimal rule separates two margins of delegation design: how often to allow the agent to depart from a business-as-usual default, and how much discretion to grant when departure is allowed. It governs these margins through a common default, a gap, and a higher discretionary interval with random boundaries. We show that the main results extend to settings in which the principal is also uncertain about the degree of preference misalignment. We further allow misalignment to vary with the state. In this case, the default may itself be an interval that grants limited discretion, while additional discretion is granted at random.

\clearpage
\appendix
\section{Appendix: Details of the Optimal Stochastic Mechanism}
\label{app:mechanism-details}

\subsection{The deterministic benchmark}

\begin{proof}[Proof of Proposition~\ref{prop:detvalue}]
Fix a deterministic delegation set $D$ and its induced action rule $a_D$.
If $a_D$ is constant over states, say $a_D(\theta)=x$, then
$
\sup_{\theta\in[0,1]}(x-\theta)^2
=\max\{x^2,(1-x)^2\}
\geq\frac14$.

Suppose that $a_D$ is nonconstant. We consider two cases according to whether the delegation set contains the agent’s ideal action at some state. If there exists $x\in D\cap[b,1+b]$, at state $\theta=x-b$, the agent chooses his ideal
action $x$, so the principal's regret is $b^2$. If
$D\cap[b,1+b]=\varnothing$, the nonconstant action rule implies that some state chooses an
action $y>1+b$; otherwise, every state would choose the same largest
available action below $b$. For that state, $y-\theta\ge y-1>b$, so the regret is strictly greater than $b^2$. Therefore, every deterministic
rule has worst-case regret at least $\min\{b^2,1/4\}$.

The bound is achievable by either full delegation or no delegation. Under $D=[0,1]$, $a_D(\theta)=\min\{\theta+b,1\}$, and the worst-case regret is $b^2$. Under
$D=\{1/2\}$, the worst-case regret is $1/4$. Hence, choosing $D=[0,1]$ for
$b\leq1/2$ and $D=\{1/2\}$ for $b\geq1/2$ attains
$\min\{b^2,1/4\}$.
\end{proof}

\subsection{Incentive compatibility}
\begin{proof}[Proof of Lemma~\ref{lem:ICmoments}]
Let $s=\mu^2+\tau$ denote the second moment. The payoff of type $\theta$
from reporting $\hat\theta$ is
$
U(\hat\theta\mid\theta)
=-s(\hat\theta)+2(\theta+b)\mu(\hat\theta)-(\theta+b)^2$. 
Adding the incentive-compatibility constraints for any $\theta$ and $\hat\theta$
gives $(\theta-\hat\theta)
(\mu(\theta)-\mu(\hat\theta))\geq0$, so $\mu$ is nondecreasing.
Moreover, $\frac{\partial}{\partial\theta}U(\hat\theta\mid\theta)
=2\bigl(\mu(\hat\theta)-\theta-b\bigr)$. The integral envelope theorem
therefore yields~\eqref{eq:IC-envelope} for the truthful payoff
$U(\theta):=U(\theta\mid\theta)$.

Conversely, if $\mu$ is nondecreasing and \eqref{eq:IC-envelope}
holds, then for any $\theta$, $\hat{\theta}$,
$
U(\theta)-U(\hat\theta\mid\theta)
=2\int_{\hat\theta}^{\theta}
\bigl(\mu(x)-\mu(\hat\theta)\bigr)\,\dd x
\geq0$. 
This establishes incentive compatibility. Finally,
as the action space is $\mathbb R$, any nonnegative variance
$\tau(\theta)$ can be realized by assigning equal probabilities to
$\mu(\theta)-\sqrt{\tau(\theta)}$ and
$\mu(\theta)+\sqrt{\tau(\theta)}$. This proves Lemma~\ref{lem:ICmoments}.
\end{proof}

We now derive some useful identities from incentive compatibility. Rearranging the envelope condition gives the identity
\begin{equation}
\label{eq:second-moment-envelope}
s(\theta)=s(0)+2(\theta+b)\mu(\theta)-2b\mu(0)
-2\int_0^\theta\mu(x)\,\dd x.
\end{equation}
Thus, for a given mean schedule, the second moment and variance schedules are uniquely determined by the initial second moment $s(0)$. The regret identity is
\begin{equation}
\label{eq:regret-envelope-integrated}
R(M,\theta)=R(M,0)+2b\bigl(\mu(\theta)-\mu(0)\bigr)-2\int_0^\theta\mu(x)\,\dd x+\theta^2.
\end{equation}
For absolutely continuous moments, differentiating these identities yields the variance identity~\eqref{eq:var_derivative} and the regret identity~\eqref{eq:regret-derivative}.

\subsection{The candidate and its worst-case distribution}
\begin{proof}[Proof of Proposition~\ref{prop:optimal mech}: existence and optimality]
Let $\alpha=e^{1-1/b}$ and $a_0=b(1-\alpha/2)$. The moments in Proposition~\ref{prop:optimal mech} are
\begin{equation}
\label{eq:optimal-mean}
\mu^*(\theta)=
\begin{cases}
\displaystyle \theta+b-\frac b2e^{(\theta+b-1)/b},&0\leq\theta\leq1-b,\\[2mm]
\displaystyle 1-\frac b2,&1-b\leq\theta\leq1,
\end{cases}
\end{equation}
\begin{equation}
\label{eq:optimal-variance}
\tau^*(\theta)=
\begin{cases}
\displaystyle V^*-\left(b-\frac b2e^{(\theta+b-1)/b}\right)^2,
&0\leq\theta\leq1-b,\\[2mm]
\displaystyle V^*-\frac{b^2}{4},&1-b\leq\theta\leq1.
\end{cases}
\end{equation}
$\mu^*$ is continuous and nondecreasing, $\tau^*$ is nonnegative, and they satisfy~\eqref{eq:IC-envelope}. This pair of mean and variance is therefore feasible and incentive compatible. The corresponding second moment and regret are
\begin{equation}
\label{eq:optimal-second-moment}
s^*(\theta)=
\begin{cases}
V^*+2\theta\mu^*(\theta)-\theta^2,&0\leq\theta\leq1-b,\\
V^*+1-b,&1-b\leq\theta\leq1,
\end{cases}
\end{equation}
\begin{equation}
\label{eq:optimal-regret}
R(M^*,\theta)=
\begin{cases}
V^*,&0\leq\theta\leq1-b,\\
V^*-(\theta-(1-b))(1-\theta),&1-b\leq\theta\leq1.
\end{cases}
\end{equation}

Define a distribution over states by 
\begin{equation*}
F^*=\frac\alpha2\delta_0
+\frac{1-\alpha/2}{b} e^{-\theta/b}\mathbf{1}_{(0,1-b)}(\theta)\,\dd\theta
+\alpha(1-\alpha/2)\delta_1.
\end{equation*}
The continuous component has mass $(1-\alpha/2)(1-\alpha)$, so together with the two endpoint atoms, $F^*$ has total mass one.

To prove \((M^*,F^*)\) 
is a saddle point, we first show that for any incentive-compatible mechanism $M$, $R(M,F^*)\geq V^*$. Integrating~\eqref{eq:regret-envelope-integrated} against $F^*$ and using $R(M,0)=s(0)$ yields
\[
\begin{aligned}
R(M,F^*)={}&s(0)-2b(1-\alpha/2)\mu(0)+2\alpha(1-\alpha/2)\int_{1-b}^1
\bigl(\mu(1)-\mu(\theta)\bigr)\,\dd\theta
+\E_{F^*}[\theta^2].
\end{aligned}
\]
A direct calculation gives
$
\E_{F^*}[\theta^2]=2b^2(1-\alpha/2)^2=2V^*$. 
Therefore, using $s(0)=\mu(0)^2+\tau(0)$ and
$a_0=b(1-\alpha/2)$,
\[
\begin{aligned}
R(M,F^*)-V^*
={}&\bigl(\mu(0)-a_0\bigr)^2+\tau(0)+2\alpha(1-\alpha/2)\int_{1-b}^1
\bigl(\mu(1)-\mu(\theta)\bigr)\,\dd\theta.
\end{aligned}
\]
This is~\eqref{eq:saddle-certificate}. Each term on the right-hand side is nonnegative: the first is a square, the second is a variance, and the third is nonnegative because incentive compatibility implies that $\mu$ is nondecreasing. Hence,
$
R(M,F^*)\geq V^*$ 
for every incentive-compatible mechanism $M$. On the other hand, \eqref{eq:optimal-regret} implies that $R(M^*,F)\leq V^*$ for every $F$. As $F^*$ assigns probability only to states at which $R(M^*,\theta)=V^*$, we also have $R(M^*,F^*)=V^*$. Therefore,
$
R(M^*,F)\leq R(M^*,F^*)=V^*\leq R(M,F^*)$  
for every $F$ and every incentive-compatible mechanism $M$. This proves \((M^*,F^*)\) 
is a saddle point.

Since $0<\alpha<1$, we have $b/2<a_0<b$. Moreover, because $1/b-1>0$, the inequality $1-e^{-(1/b-1)}<\frac{1}{b}-1$ implies $2b-be^{1-1/b}<1$, and hence $a_0<1/2$. Therefore, $V^*<\min\{b^2,1/4\}$.
\end{proof}

\subsection{Uniqueness of the moments}

\begin{proof}[Proof of Proposition~\ref{prop:optimal mech}: uniqueness]
Let $\widetilde M$ be any minimax-optimal mechanism. The lower
bound established above and minimax optimality imply
\[
V^*
\le R(\widetilde M,F^*)
\le \sup_{\theta\in\Theta}R(\widetilde M,\theta)
=V^*.
\]
Hence $R(\widetilde M,F^*)=V^*$. Denote the conditional mean
and variance induced by $\widetilde M$ by $(\mu,\tau)$.
The lower-bound identity then gives
\[
0
=
\bigl(\mu(0)-a_0\bigr)^2+\tau(0)
+2F^*(\{1\})\int_{1-b}^1
\bigl(\mu(1)-\mu(x)\bigr)\,\dd x.
\]
Each term on the right-hand side is nonnegative, since
$\tau(0)\geq0$ and incentive compatibility implies that
$\mu$ is nondecreasing. Moreover, $F^*(\{1\})>0$. Therefore,
\[
\mu(0)=a_0,\qquad
\tau(0)=0,\qquad
\int_{1-b}^1\bigl(\mu(1)-\mu(x)\bigr)\,\dd x=0.
\]

Since $R(\widetilde M,\theta)\leq V^*$ at every state and
$R(\widetilde M,F^*)=V^*$, regret equals $V^*$ for $F^*$-almost every
state. Because $F^*$ has a positive density on $(0,1-b)$,
\eqref{eq:regret-envelope-integrated} implies that, for
Lebesgue-almost every $\theta$ in this interval,
$
2b\bigl(\mu(\theta)-a_0\bigr)
-2\int_0^\theta\mu(x)\,\dd x+\theta^2=0$. 
This identity implies that $\mu$ agrees almost everywhere with an
absolutely continuous function. Constant regret
and~\eqref{eq:regret-derivative} then give
$
b\mu'(\theta)=\mu(\theta)-\theta$ 
almost everywhere on $(0,1-b)$. Together with $\mu(0)=a_0$, this
differential equation has the unique solution given by the first
branch of~\eqref{eq:optimal-mean}. Since $\mu$ is nondecreasing and the
solution is continuous, agreement almost everywhere implies that
$\mu=\mu^*$ throughout $(0,1-b)$.

The vanishing upper-region integral and the monotonicity of $\mu$
imply that $\mu(\theta)=\mu(1)=:m$ for every $\theta\in(1-b,1]$.
Monotonicity gives $m\geq\mu^*(1-b)$. Taking the right limit
of~\eqref{eq:regret-envelope-integrated} at $1-b$ gives
$
R\bigl(\widetilde M,(1-b)^+\bigr)
=V^*+2b\bigl(m-\mu^*(1-b)\bigr)$. 
The worst-case bound therefore forces
$m=\mu^*(1-b)=1-b/2$. Monotonicity also pins down the value at $1-b$,
so $\mu=\mu^*$ at every state. Finally,
\eqref{eq:second-moment-envelope} and $s(0)=a_0^2$ imply that
$s=s^*$ and hence $\tau=\tau^*$.
\end{proof}

\begin{cro}
\label{cor:randomizationnecessary}
Every optimal stochastic mechanism induces a nondegenerate action lottery at every state $\theta>0$ and the deterministic action $a_0$ at $\theta=0$.
\end{cro}

\begin{proof}
Proposition~\ref{prop:optimal mech} gives $\tau^*(0)=0$. For $0<\theta\leq1-b$, the upward distortion in~\eqref{eq:optimal-variance} is strictly smaller than its value $a_0$ at zero, so $\tau^*(\theta)>0$. For $\theta\geq1-b$, $\tau^*(\theta)=a_0^2-b^2/4>0$ because $a_0>b/2$. Since every optimal mechanism has these moments, its action lottery is degenerate exactly at $\theta=0$.
\end{proof}

\section{Appendix: Construction of the Random Delegation Rule}
\label{app:endpoint-construction}

We construct the endpoint distributions used in Section~\ref{subsec:implementation-random-delegation}. Use $t=\theta+b$ for the agent's ideal action and define
\[
\lambda(t)=1-\frac12e^{(t-1)/b},\qquad b\leq t\leq1.
\]
Thus, $\lambda(t)=\mu^{*\prime}(t-b)$ for $b<t<1$. Note that $\lambda(\cdot)$ measures how the optimal mean action responds locally to a change in the agent’s ideal action. Recall that the smallest admissible lower bound of the discretionary interval is $\underline a_{\min}:=2b-a_0.$

\subsection{Endpoint distributions}

We start with the construction of $G$. Set $G(x)=0$ for $x\leq\underline a_{\min}$. Starting at $\underline a_{\min}$, define $G$ by the delay equation
\begin{equation}
\label{eq:floor-construction}
G'(x)=
\frac{\lambda\!\left((a_0+x)/2\right)-G\!\left((a_0+x)/2\right)}
{2(x-a_0)}.
\end{equation}
Since $(a_0+x)/2<x$ and the denominator is bounded away from zero on the construction interval, the equation determines $G$ successively from already specified values. We use the delay equation to determine $G'$ as long as it
remains positive. Accordingly, define
$\underline a_{\max}:=\inf\{x\in(\underline a_{\min},2-a_0):
G'(x)=0\}$ if the set is nonempty, and set
$\underline a_{\max}=2-a_0$ otherwise. The resulting distribution of $G$ is
\[
G(x)=
\begin{cases}
0,
& x\le \underline a_{\min},\\[3mm]
\displaystyle
\int_{\underline a_{\min}}^x
\frac{
\lambda\bigl((a_0+u)/2\bigr)-G\bigl((a_0+u)/2\bigr)
}{2(u-a_0)}\,\mathrm{d}u,
& \underline a_{\min}<x<\underline a_{\max},\\[4mm]
\displaystyle
\int_{\underline a_{\min}}^{\underline a_{\max}}
\frac{
\lambda\bigl((a_0+u)/2\bigr)-G\bigl((a_0+u)/2\bigr)
}{2(u-a_0)}\,\mathrm{d}u,
& x\ge \underline a_{\max}.
\end{cases}
\]
Before stopping, the integrand is nonnegative and bounded above by $C:=\frac{\lambda(b)}{4(b-a_0)}<\infty.$ Consequently, for any $y<z$, the integral representation and the constant extensions give
\[
0\le G(z)-G(y)\le C(z-y).
\]
Thus, $G$ is Lipschitz continuous and nondecreasing on
$\mathbb{R}$, with a strictly positive density on
$(\underline a_{\min},\underline a_{\max})$. In particular, the lower-endpoint subdistribution has no atoms.

The upper endpoint $\underline a_{\max}$ determines two quantities used below: the highest threshold in the agent's ideal action for switching from $a_0$ to a floor among the delegation sets, and the total mass of $G$. Define
\[
c=\frac{a_0+\underline a_{\max}}{2}\leq1,
\qquad q=G(\underline a_{\max}).
\]
For any state $\theta>c-b$, we have $2(\theta+b)-a_0>\underline a_{\max}$, so no further jump from the default $a_0$ to the lower endpoint can occur in any delegation set. If $c<1$, the stopping condition gives $G(c)=\lambda(c)$. This means that at $c$, the mass of delegation sets whose chosen action moves with the agent's ideal action is sufficient to generate the required marginal change in the mean.

Now, we construct the marginal subdistribution of the upper endpoints. Set
\begin{equation}
\label{eq:ceiling-construction}
H(x)=
\begin{cases}
0,&x<c,\\[1mm]
G(x)-\lambda(x),&c\leq x<1,\\[1mm]
G(x),&x\geq1.
\end{cases}
\end{equation}
When $c=1$, the middle region is empty. The three branches are chosen to match the optimal mean response. For \(b<t<c\), the delay equation already gives \(\lambda(t)=G(t)+4(t-a_0)G'(2t-a_0)\), so \(H(t)=0\). When $c<t<1$, $G(c)=\lambda(c)$ implies that the choices following the agent's ideal action generate exactly the required marginal change in the mean at $c$. As $t$ rises above $c$, $G(t)$ weakly increases while
$\lambda(t)$ decreases. The upper endpoints slow the growth of the mean:
once $t$ passes an upper endpoint, the corresponding choice stops rising
with $t$. Since $H(t)$ is the mass of these capped choices, the mean's
marginal response is $G(t)-H(t)$. Matching $\lambda(t)$ therefore gives
$H(t)=G(t)-\lambda(t)$ for $c<t<1$. At $t=1$, $H(t)$ jumps by $1/2$ if $c<1$, and by $G(1)$ if $c=1$. For \(t>1\), the optimal mean is constant, so \(H(t)=G(t)\).

Then, we verify $q<1$, so that the marginal subdistributions are
well-defined. Let $\lambda_0=\lambda(b)<1$ and $W(x)=(x-a_0)\bigl(\lambda_0-G(x)\bigr).$ Equation~\eqref{eq:floor-construction} implies, for
$x\in[\underline a_{\min},\underline a_{\max}]$,
\begin{equation}
\label{eq:endpoint-mass-bound}
\begin{split}
W(x)={}&\lambda_0(b-a_0)
+\int_{(a_0+x)/2}^x\frac{W(v)}{v-a_0}\,\dd v+\frac12\int_{\underline a_{\min}}^x
\left[\lambda_0-\lambda\!\left(\frac{a_0+v}{2}\right)\right]\dd v.
\end{split}
\end{equation}
This identity follows by differentiating both sides and checking equality
at $\underline a_{\min}$. Initially, $G=0$ on
$[b,\underline a_{\min}]$, so $W>0$ there. At any putative first zero
of $W$, the first term on the right side
of~\eqref{eq:endpoint-mass-bound} is strictly positive, the second is
positive, and the third is nonnegative. Such a zero is impossible.
Hence, $G(x)<\lambda_0$ throughout the construction interval, and
$0<q<\lambda_0<1$.

Finally, couple the lower and upper endpoints by using the same quantile. Draw \(Z\) uniformly from \((0,1)\), and set
\[
\underline a=\inf\{x:G(x)\geq qZ\},
\qquad
\bar a=\inf\{x:H(x)\geq qZ\}.
\]
Since \(H\leq G\), this gives \(\underline a\leq\bar a\). Define \(\Psi^*\) to draw the delegation set \(\{a_0\}\cup[\underline a,\bar a]\) with probability \(q\), and the singleton \(\{a_0\}\) with probability \(1-q\).

\subsection{Matching the moments}

\begin{proof}[Proof of Proposition~\ref{thm:delegation-implementation}]
Let $\mu$ denote the mean induced by $\Psi^*$ and write $t=\theta+b$. Every realized delegation set induces $a_0$ at $\theta=0$, so $\mu(0)=a_0=\mu^*(0).$ By~\eqref{eq:endpoint-response}, the induced mean satisfies
\[
\mu'(\theta)
=G(t)-H(t)+4(t-a_0)G'(2t-a_0)
\]
almost everywhere. The reason is that for a realized set $D=\{a_0\}\cup[\underline a,\bar a]$, the agent initially chooses $a_0$, switches to $\underline a$ when $t$ crosses $(a_0+\underline a)/2$, follows his ideal action $t$ while it lies in $[\underline a,\bar a]$, and chooses $\bar a$ once the upper bound binds. Consequently, a marginal increase in $\theta$ changes the induced mean through three components. The mass $G(t)$ has a lower endpoint below $t$ and can adjust with the ideal action, while the mass $H(t)$ is already constrained by its upper endpoint and must be subtracted. In addition, realizations with $\underline a=2t-a_0$ switch the agent's choice from $a_0$ to $\underline a$. When $\theta$ increases by $\dd\theta$, the switching lower endpoint
$2t-a_0$ increases by $2\dd\theta$. The probability mass that switches is therefore
$2G'(2t-a_0)\dd\theta$. Each such realization raises the chosen action by $2(t-a_0)$. Hence, these switches contribute $4(t-a_0)G'(2t-a_0)$ to $\mu'(\theta)$.

For $b<t<c$, $H(t)=0$, and~\eqref{eq:floor-construction}, evaluated at $x=2t-a_0$, gives $4(t-a_0)G'(2t-a_0)=\lambda(t)-G(t).$ Therefore, $\mu'(\theta)=\lambda(t)=\mu^{*\prime}(\theta).$ 

For $c<t<1$, we have $2t-a_0>\underline a_{\max}$, so $G'(2t-a_0)=0$, while~\eqref{eq:ceiling-construction} gives $G(t)-H(t)=\lambda(t)$. Again, we have $\mu'(\theta)=\lambda(t)=\mu^{*\prime}(\theta).$

Finally, for $1<t<1+b$, $H(t)=G(t)$ and $2t-a_0>\underline a_{\max}$, so $\mu'(\theta)=0=\mu^{*\prime}(\theta).$

Thus, $\mu'=\mu^{*\prime}$ almost everywhere on $(0,1)$. Since both mean schedules are absolutely continuous and satisfy $\mu(0)=\mu^*(0)$, it follows that
$
\mu(\theta)=\mu^*(\theta),\forall \theta\in[0,1]$.
We next verify the second moment. Let
$s(\theta)=\E_{\Psi^*}[a_D(\theta)^2]$. In realizations with $\underline a\leq t<\bar a$,
the chosen action equals $t$, so its square increases at rate $2t$.
These realizations have probability $G(t)-H(t)$. A switch from $a_0$
to $\underline a$ occurs when $\underline a=2t-a_0$. At the switch,
the squared action increases by
$
(2t-a_0)^2-a_0^2=4t(t-a_0)$. 
As $\theta$ increases, the switching threshold $2t-a_0$ moves at
rate $2$. Thus, almost everywhere,
\[
\begin{aligned}
s'(\theta)
&=2t\bigl[G(t)-H(t)\bigr]
 +8t(t-a_0)G'(2t-a_0)\\
&=2t\bigl[G(t)-H(t)+4(t-a_0)G'(2t-a_0)\bigr]=2t\mu'(\theta).
\end{aligned}
\]
For $0<\theta<1-b$, equations~\eqref{eq:optimal-mean}
and~\eqref{eq:optimal-second-moment} give
\[
\begin{aligned}
s^{*\prime}(\theta)
&=2\bigl(\mu^*(\theta)-\theta\bigr)
  +2\theta\mu^{*\prime}(\theta)=2(\theta+b)\mu^{*\prime}(\theta),
\end{aligned}
\]
where the second equality uses
$\mu^*(\theta)-\theta=b\mu^{*\prime}(\theta)$.
For $1-b<\theta<1$, both $\mu^*$ and $s^*$ are constant, so the
same identity holds. Since $\mu=\mu^*$, we have
$s'=s^{*\prime}$ almost everywhere. Every realized delegation set
induces $a_0$ at $\theta=0$, giving
$s(0)=a_0^2=s^*(0)$. Hence $s=s^*$ at every state, and therefore
$\tau=\tau^*$.

It follows that $\Psi^*$ induces the same statewise regret as $M^*$.
Since every random delegation rule induces an incentive-compatible
stochastic mechanism, the saddle-point inequalities for $(M^*,F^*)$
also hold for $(\Psi^*,F^*)$. This proves
Proposition~\ref{thm:delegation-implementation}.
\end{proof}

The construction uses only bounded actions: $\underline a\leq\underline a_{\max}\leq2-a_0,$ and $\bar a\leq\max\{1,\underline a_{\max}\}\leq2-a_0.$ Thus, every realized delegation set lies in $[a_0,2-a_0]$.

\begin{rmk}[Tie-breaking]
\label{rmk:tie-breaking-saddle}For any fixed state, a tie between the default and the lower endpoint requires
$
\underline a=2(\theta+b)-a_0$. 
The constructed $G$ is continuous, so this event has probability zero. The induced moments and regret of $\Psi^*$ are therefore unchanged by any measurable tie-breaking rule. For competing delegation rules, principal-favoring tie-breaking minimizes regret among the agent's best responses. Any alternative rule weakly increases their regret against $F^*$. Thus, $(\Psi^*,F^*)$ remains a saddle point under any such tie-breaking convention.
\end{rmk}

\subsection{Proofs of the structural results}

\begin{proof}[Proof of Proposition~\ref{prop:necessary-default-gap}]
Proposition~\ref{thm:delegation-implementation} implies $V_D=V^*$. Hence, any optimal delegation rule induces an optimal stochastic mechanism. By Proposition~\ref{prop:optimal mech}, its action at $\theta=0$ has mean $a_0$ and zero variance, so $a_D(0)=a_0$ for almost every realized set $D$. In particular, $a_0\in D$ almost surely. The lowest-state agent's ideal action is $b$. Every action $a\in(a_0,2b-a_0)$ is strictly closer to $b$ than $a_0$ is. If such an action belonged to a realized set, the lowest-state agent would not choose $a_0$, a contradiction. Thus, every optimal rule contains the common default and excludes $(a_0,2b-a_0)$ almost surely.
\end{proof}

\begin{proof}[Proof of Theorem~\ref{thm:optimal-delegation}]
The endpoint construction in Appendix~\ref{app:endpoint-construction} defines a probability distribution $\Psi^*$ supported on $\{a_0\}$ and on sets $\{a_0\}\cup[\underline a,\bar a]$ with $\underline a\geq2b-a_0$. Proposition~\ref{thm:delegation-implementation} shows that this rule induces the optimal moments $(\mu^*,\tau^*)$ and attains $V_D=V^*$. It is therefore optimal both among random delegation rules and relative to the larger class of stochastic incentive-compatible mechanisms. Proposition~\ref{prop:necessary-default-gap} establishes the stated common default and gap as necessary features of every optimum.
\end{proof}

\subsection{Proof of the Sequential Implementation}
\label{app:sequential-implementation}

\begin{proof}[Proof of Proposition~\ref{prop:sequential-implementation}]
Let $q$ be the probability that $\Psi^*$ draws a discretionary
interval. Draw potential endpoints \((\underline a,\bar a)\) according to the joint endpoint distribution among interval-containing sets under the original rule \(\Psi^*\). Independently draw a Bernoulli variable \(B\) with success probability \(q\). All these draws are independent
of the state. Under the static rule, they generate the delegation set
\[
D=
\begin{cases}
\{a_0\}, & B=0,\\
\{a_0\}\cup[\underline a,\bar a], & B=1.
\end{cases}
\]
This set has distribution $\Psi^*$. In particular, when $B=0$ the
static set is $\{a_0\}$, regardless of the drawn endpoints.

In the sequential procedure, announce \(\underline a\) regardless of \(B\), before the agent decides whether to retain \(a_0\) or continue, while keeping \(B\) unannounced. Retaining the default implements $a_0$
regardless of $B$. Following continuation, implement $a_0$ if
$B=0$; if $B=1$, reveal $[\underline a,\bar a]$ and let the agent
choose from it. Thus the announced floor has the same distribution
whether $B=0$ or $B=1$, and observing it does not reveal whether
discretion will be granted.

Write $t=\theta+b$ for the agent's ideal action. We show that the lower endpoint $\underline a$ determines whether the agent
prefers a discretionary interval to the default $a_0$.

If $t\le\underline a$, the agent's preferred action in the interval
is $\underline a$. The payoff gain relative to the default is
$
(a_0-t)^2-(\underline a-t)^2
=(\underline a-a_0)(2t-a_0-\underline a)$. 
Since $\underline a>a_0$, this gain is positive when $t>(a_0+\underline a)/2$.

If $t>\underline a$, the floor is closer to the agent's ideal action than the default:
\[
\min_{a\in[\underline a,\bar a]}(a-t)^2
\le(\underline a-t)^2
<(a_0-t)^2.
\]
Thus the interval is strictly preferred in this case as well.
Combining the two cases, the agent prefers the
interval exactly when $t>(a_0+\underline a)/2$, regardless of
$\bar a$.

After observing $\underline a$, the agent knows that an interval
will be offered with probability $q$, since $B$ is independent of
$(t,\underline a)$. If $B=0$, continuing and retaining the default
both implement $a_0$. If $B=1$, the payoff gain from continuing has
the same sign for every possible ceiling, as shown above. Since
$q>0$, the agent therefore continues exactly when
$t>(a_0+\underline a)/2$. 

Finally, fix any state and any realization of
$(\underline a,\bar a,B)$. If $B=0$, both procedures implement
$a_0$, whether or not the agent continues in the sequential
procedure. If $B=1$ and $t\le(a_0+\underline a)/2$, the agent chooses
$a_0$ from the static set and retains it in the sequential
procedure. If $B=1$ and $t>(a_0+\underline a)/2$, both procedures
implement the agent's preferred action in the interval,
$
\min\{\bar a,\max\{\underline a,t\}\}.$ 
Thus the two procedures implement the same action for every state
and every realization of the draws. They therefore induce the same
statewise action distribution and worst-case regret.
\end{proof}

\section{Appendix: Robustness to Uncertain Bias}
\label{app:uncertain-bias}

\begin{proof}[Proof of Proposition~\ref{prop:uncertain-bias}]
We first establish the lower bound. Fix any random delegation rule $\Psi$. The ambiguity set $\widetilde{\cF}_b$ allows nature to set $\widetilde b=b$ almost surely and draw $\theta$ from the baseline least-favorable distribution $F^*$. Under this distribution, the baseline saddle-point inequality implies that $\Psi$ yields expected regret at least $V^*$. Its worst-case regret is therefore also at least $V^*$.

For the upper bound, we show that a smaller bias weakly reduces the principal's regret under every fixed delegation set. Fix $\theta\in[0,1]$, $\widetilde b\in[0,b]$, and a delegation
set $D$. Write
$
x=a_D(\theta+\widetilde b)$, and $
y=a_D(\theta+b)$. 
As the agent's choice rule is nondecreasing in the agent's ideal
point, $x\leq y$. Since the agent with bias $\widetilde b$
chooses $x$ when $y$ is also available, we have $(x-\theta-\widetilde b)^2
\leq
(y-\theta-\widetilde b)^2.$ Expanding gives
$
(y-\theta)^2-(x-\theta)^2
\geq2\widetilde b(y-x)\geq0.$ 
This comparison holds for every realized delegation set. Averaging over
$D\sim\Psi^*$ therefore yields
\begin{equation*}
\begin{aligned}
\E_{\Psi^*}\!\left[
\bigl(a_D(\theta+\widetilde b)-\theta\bigr)^2
\right]
\leq
\E_{\Psi^*}\!\left[
\bigl(a_D(\theta+b)-\theta\bigr)^2
\right]
\leq V^*.
\end{aligned}
\end{equation*}
Integrating this state-wise bound against any joint distribution
$\widetilde F\in\widetilde{\cF}_b$ shows that the worst-case regret
under $\Psi^*$ is at most $V^*$. Together with the lower bound,
this establishes that $\Psi^*$ solves
\eqref{eq:uncertain-bias-problem} with value $V^*$, proving the
proposition.
\end{proof}

\paragraph{Optimality with two-dimensional private information.}
The minimax value remains $V^*$ if the agent privately observes both $\theta$ and $\widetilde b$ and the principal may choose any incentive-compatible stochastic mechanism for this two-dimensional type. For the lower bound, the restriction of any such mechanism to types with $\widetilde b=b$ is incentive compatible in the baseline model. Nature can therefore choose to set $\widetilde b=b$ almost surely and draw $\theta$ from the baseline least-favorable distribution $F^*$, under which the baseline saddle-point inequality gives expected regret at least $V^*$. For the upper bound, assign to each reported type $(\theta,\widetilde b)$ the distribution of $a_D(\theta+\widetilde b)$ under $D\sim\Psi^*$. This direct mechanism is incentive compatible: for every realized set $D$, truthful reporting selects an action closest to the agent's true ideal point, whereas any misreport selects another action in the same set. Taking expectations over $D$ preserves this incentive inequality. The pointwise regret bound established above gives the matching upper bound $V^*$.

\section{Appendix: State-Dependent Bias}
\label{app:general-bias}

\subsection{The optimal stochastic mechanism}
We first characterize the optimal stochastic mechanism under the state-dependent bias specification of Section~\ref{subsec:state-dependent-bias}. Define the upper cutoff by $\bar\theta=y_A^{-1}(1).$

For $0\leq\theta\leq\bar\theta$ with $b(\theta)>0$, define the distortion
\begin{equation}
\label{eq:general-distance-solution}
d(\theta)
=\frac{b(\bar\theta)}{2}
\exp\!\left(-\int_\theta^{\bar\theta}\frac{\mathrm dx}{b(x)}\right)
+\int_\theta^{\bar\theta}
\exp\!\left(-\int_\theta^t\frac{\mathrm dx}{b(x)}\right)\mathrm dt.
\end{equation}
This function satisfies
\begin{equation}
\label{eq:general-distance-ode}
d'(\theta)=\frac{d(\theta)}{b(\theta)}-1,
\qquad
d(\bar\theta)=\frac{b(\bar\theta)}{2},
\end{equation}
where the differential equation applies on the interior where $b>0$. If $b(1)=0$, then $\bar\theta=1$; interpret the integrals in~\eqref{eq:general-distance-solution} as improper integrals, with $\exp(-\infty)=0$, and set $d(1)=0$. Define the lower cutoff by
\begin{equation*}
\underline\theta
=\inf\bigl\{\theta\in[0,\bar\theta]:
b(\theta)>0,\ d(\theta)\leq b(\theta)\bigr\}.
\end{equation*}

\begin{prop}[Optimal stochastic mechanism under state-dependent bias]
\label{prop:general-optimal-mechanism}
Let $V=d(\underline\theta)^2$. There exists an optimal incentive-compatible stochastic mechanism with mean and variance
\begin{equation}
\label{eq:general-optimal-moments}
\bigl(\mu^*(\theta),\tau^*(\theta)\bigr)
=
\begin{cases}
\bigl(y_A(\theta),0\bigr),
&0\leq\theta<\underline\theta,\\[3pt]
\bigl(\theta+d(\theta),V-d(\theta)^2\bigr),
&\underline\theta\leq\theta\leq\bar\theta,\\[3pt]
\left(\dfrac{1+\bar\theta}{2},V-\dfrac{b(\bar\theta)^2}{4}\right),
&\bar\theta<\theta\leq1.
\end{cases}
\end{equation}
The unrestricted minimax value is $V^*=V$, and the action at state zero is deterministic, with $a_0=\mu^*(0)$.
\end{prop}

Below $\underline\theta$, the mechanism accommodates the agent by assigning his ideal action without randomization. On $[\underline\theta,\bar\theta]$, the mean responds to the state and the variance keeps regret constant at $V$. Above $\bar\theta$, both moments are constant. The accommodation region is empty when $\underline\theta=0$, and the pooling region is empty when $\bar\theta=1$.

\begin{proof}[Proof of Proposition~\ref{prop:general-optimal-mechanism}]
The argument for Lemma~\ref{lem:ICmoments}, with $y_A(\theta)$ in place of $\theta+b$, gives the following characterization. A mechanism is incentive compatible if and only if its mean $\mu$ is nondecreasing and, writing $U(\theta)=-\tau(\theta)-[\mu(\theta)-y_A(\theta)]^2$, its moments satisfy
\begin{equation}
\label{eq:general-IC-envelope}
U(\hat\theta)-U(\theta)
=2\int_{\theta}^{\hat\theta}
\bigl[\mu(x)-y_A(x)\bigr]y_A'(x)\,\mathrm dx
\qquad(0\leq\theta\leq\hat\theta\leq1).
\end{equation}

We first verify feasibility of the proposed moments. The integral representation of $d$ gives $d(\theta)>0$ wherever $b(\theta)>0$ and $\theta<\bar\theta$. On this interval,
\[
d'(\theta)=\frac{d(\theta)}{b(\theta)}-1,
\qquad
(d-b)'(\theta)=\frac{d(\theta)-b(\theta)}{b(\theta)}-b'(\theta).
\]
If the bias is weakly increasing, an integrating factor in the second equation shows that $d\leq b$ continues to hold after any state at which it holds. The definition of $\underline\theta$ therefore gives $d\leq b$ throughout the responsive region. If the bias is weakly decreasing, $d(\theta)>b(\theta)$ at any state would imply that $d-b$ remains positive and increasing thereafter. This contradicts the terminal condition, including when $d(1)=b(1)=0$. Thus $d\leq b$ in this case as well, and $\underline\theta=0$.

Consequently, $d$ is nonincreasing on the responsive region. If $b(0)=0$, then $\underline\theta>0$: otherwise $0<d(\bar\theta)\leq d(\theta)\leq b(\theta)\to0$ as $\theta\downarrow0$, a contradiction. In particular, $b(\underline\theta)>0$. When $\underline\theta>0$, continuity and the cutoff definition imply $d(\underline\theta)=b(\underline\theta)$. The moments therefore join continuously to deterministic accommodation. They also join continuously to pooling, since $\bar\theta+d(\bar\theta)=(1+\bar\theta)/2.$

If $b(1)=0$, the integral representation instead gives $0\leq d(\theta)\leq1-\theta$, so continuity holds at $\bar\theta=1$.

On the interior of the responsive region, where $b>0$, the proposed moments satisfy
\[
\mu^{*\prime}(\theta)=\frac{d(\theta)}{b(\theta)}\in[0,1],
\qquad
\tau^{*\prime}(\theta)
=2\bigl[b(\theta)-d(\theta)\bigr]\frac{d(\theta)}{b(\theta)}
=2\bigl[y_A(\theta)-\mu^*(\theta)\bigr]\mu^{*\prime}(\theta).
\]
Since $d$ is nonincreasing, the variance is nonnegative. The bounds $0\leq d/b\leq1$ give bounded derivatives for both moments on the responsive region, including as $\theta\uparrow1$ when $b(1)=0$. Together with continuity at the cutoffs, deterministic accommodation, and constant pooling, these bounds imply that both moments are absolutely continuous on $[0,1]$. The mean is nondecreasing. The displayed identity yields~\eqref{eq:general-IC-envelope} on the responsive region. The same envelope holds under deterministic accommodation and under pooling, so continuity at the cutoffs establishes incentive compatibility on the entire state space.

Write $V=d(\underline\theta)^2$. The candidate regret equals $V$ on $[\underline\theta,\bar\theta]$. Above $\bar\theta$, the pooled mean is the midpoint of $[\bar\theta,1]$, so regret is at most $V$, with equality at state one. If $\underline\theta>0$, the bias is weakly increasing and regret below the cutoff is $b(\theta)^2\leq b(\underline\theta)^2=V$. Thus the mechanism guarantees $V$.

For the matching lower bound, define the cumulative distribution function
\begin{equation}
\label{eq:general-worst-case-cdf}
F^*(\theta)=
\begin{cases}
0,&\theta<\underline\theta,\\[3pt]
1-\dfrac{d(\underline\theta)}{b(\underline\theta)}
\exp\!\left(-\displaystyle\int_{\underline\theta}^{\min\{\theta,\bar\theta\}}
\frac{y_A'(x)}{b(x)}\,\mathrm dx\right),
&\underline\theta\leq\theta<1,\\[3pt]
1,&\theta\geq1.
\end{cases}
\end{equation}
The bounds $0<d(\underline\theta)\leq b(\underline\theta)$ and $y_A'\geq0$ make this a probability distribution. It is supported on $[\underline\theta,\bar\theta]\cup\{1\}$ and satisfies
\[
b(\theta)F^{*\prime}(\theta)
=y_A'(\theta)\bigl[1-F^*(\theta)\bigr]
\quad\text{for }\underline\theta<\theta<\bar\theta.
\]
When $\bar\theta=1$, the integral in~\eqref{eq:general-worst-case-cdf} is interpreted by its limit from below. Let $p=1-F^*(1^-)$ denote the mass at state one.

Take any incentive-compatible mechanism $M$, with moments $(\mu,\tau)$, and put $s=\mu^2+\tau$. Rearranging~\eqref{eq:general-IC-envelope} gives, for $\theta\geq\underline\theta$,
\[
\begin{aligned}
R(M,\theta)
=s(\underline\theta)-2y_A(\underline\theta)\mu(\underline\theta)
+\theta^2+2b(\theta)\mu(\theta)-2\int_{\underline\theta}^{\theta}y_A'(x)\mu(x)\,\mathrm dx.
\end{aligned}
\]
Integrate against $F^*$. The density identity cancels the terms involving $\mu$ on the interior of the responsive region. The atom at $\underline\theta$ leaves the coefficient $-2[\underline\theta+d(\underline\theta)]$ on $\mu(\underline\theta)$. Finally, $y_A(\bar\theta)=1$ gives $b(1)=\int_{\bar\theta}^1y_A'(x)\,\mathrm dx$. Completing the square therefore yields
\begin{equation}
\label{eq:general-saddle-certificate}
\begin{aligned}
R(M,F^*)
={}&C+\bigl[\mu(\underline\theta)-\mu^*(\underline\theta)\bigr]^2
+\tau(\underline\theta)+2p\int_{\bar\theta}^1
y_A'(x)\bigl[\mu(1)-\mu(x)\bigr]\,\mathrm dx,
\end{aligned}
\end{equation}
where $C=\mathbb E_{F^*}[\theta^2]-\mu^*(\underline\theta)^2$ is independent of $M$. This calculation uses the integrated envelope and hence also applies when the competing mean has jumps. For the candidate mechanism, the square, initial variance, and pooling integral all vanish. Its regret is $V$ on the support of $F^*$, so $C=V$. Every remaining term in~\eqref{eq:general-saddle-certificate} is nonnegative by monotonicity of $\mu$ and nonnegativity of the variance. Thus $R(M,F^*)\geq V$ for every incentive-compatible mechanism, proving $V^*=V$.
\end{proof}

\subsection{Implementation by random delegation}

\begin{proof}[Proof of Proposition~\ref{prop:general-state-dependent-bias}]
We construct a random delegation rule that reproduces the optimal mean
and variance at every state. The agent's expected payoff under the optimal
mechanism is
\[
U^*(\theta):=-\tau^*(\theta)-[\mu^*(\theta)-y_A(\theta)]^2.
\]
Index types by their ideal action
$t=y_A(\theta)$, with endpoints $t_0=y_A(0)$ and $t_1=y_A(1)$, and
denote the optimal mean and agent payoff by
\[
m(t)=\mu^*(y_A^{-1}(t)),\qquad
u^*(t)=U^*(y_A^{-1}(t)).
\]
By Appendix D.1, $m$ is continuous and
nondecreasing, $m(t)\le t$, and $m$ is constant on $[1,t_1]$. The lowest
type receives $a_0=m(t_0)$ with certainty. Substituting $v=y_A(x)$ in
the envelope condition from Appendix D.1 and using
$u^*(t_0)=-(a_0-t_0)^2$ gives
\begin{equation}
\label{eq:general-ideal-action-envelope}
t^2+u^*(t)=2t_0a_0-a_0^2+2\int_{t_0}^t m(v)\,dv.
\end{equation}
Adding $t^2$ removes the action-independent term from the quadratic payoff.

First, we construct a sequence of random delegation rules whose induced
means converge to $m$. For each $n$, choose a grid
$t_0=x_0<\cdots<x_N=1$, with the largest grid spacing tending to zero
as $n\to\infty$. We suppress the dependence of the grid and the draws on
$n$. Start with the delegation set $\{a_0\}$, and let $P_j$ denote the
largest available action after step $j$, with $P_0=a_0$.
At step $j$, draw $B_j\in\{0,1\}$ independently of the previous draws,
with $\Pr(B_j=1)=q_j$, and set
\[
q_j=
\frac{m(x_j)-m(x_{j-1})}
     {2\E[(x_j-P_{j-1})_+]},
\qquad
P_j=P_{j-1}+2B_j(x_j-P_{j-1})_+,
\]
where $(\cdot)_+$ denotes the positive part. If $B_j=1$ and
$P_{j-1}<x_j$, add the action
$2x_j-P_{j-1}$, the reflection of the previous largest action about $x_j$.
The type with ideal action $x_j$ is indifferent between these two actions.
In all other cases, make no addition.

Each $q_j$ is a deterministic number computed recursively from the
distribution of $P_{j-1}$. We show by induction that, for $j=1,\ldots,N$,
$q_j$ is well-defined, belongs to $[0,1/2]$, and gives
$\E P_j=m(x_j)$. Initially, $P_0=a_0=m(x_0)$. Suppose
$\E P_{j-1}=m(x_{j-1})$. Since $m(x)\le x$,
\[
\E[(x_j-P_{j-1})_+]\ge x_j-m(x_{j-1})
\ge
\max\{x_j-x_{j-1},\,m(x_j)-m(x_{j-1})\}.
\]
These bounds make the denominator strictly positive and ensure
$0\le q_j\le1/2$. Independence of $B_j$ and the previous draws gives
\[
\E P_j
=\E P_{j-1}+2q_j\E[(x_j-P_{j-1})_+]
=m(x_j).
\]

The realized draws generate the delegation set $D_n=\{P_0,\ldots,P_N\}$.
Let $\Psi_n$ denote the resulting lottery over delegation sets.
Adding the reflected action at step $j$
leaves choices below $x_j$ unchanged and makes the new action preferred
above $x_j$. The mean $m_n$ induced by $\Psi_n$ satisfies
\[
m_n(t)=m(x_{j-1})\quad\text{for }t\in(x_{j-1},x_j),
\qquad
m_n(t_0)=a_0,
\]
and $m_n(t)=m(1)$ for $t>1$. At the grid points,
$m(x_{j-1})\le m_n(x_j)\le m(x_j)$. Uniform continuity of $m$ on
$[t_0,1]$ and constancy on $[1,t_1]$ imply that $m_n$ converges
uniformly to $m$ on $[t_0,t_1]$.

Every realized set contains $a_0$ and lies in $K=[a_0,2-a_0]$, because
each added action exceeds $P_{j-1}\ge a_0$ and is at most
$2x_j-a_0\le2-a_0$. The first addition, when present, exceeds
$2t_0-a_0$, and later additions are larger, leaving the interval
$(a_0,2t_0-a_0)$ empty. Since every indifference point $x_j$ lies in
$[t_0,1]$, the largest action $P_N$ is a best response for the type
with ideal action $1$, and $\E P_N=m(1)$.

We next take a limit of these random delegation rules. The space of
nonempty compact subsets of $K$ is compact under the Hausdorff metric.
By Prokhorov's theorem, a subsequence $\Psi_{n_k}$ converges weakly to
a lottery $\Psi$ over compact delegation sets. Let $\mathcal C$ be the
family of nonempty compact subsets $D$ of $K$ satisfying
\[
a_0\in D,\qquad
D\cap(a_0,2t_0-a_0)=\varnothing,\qquad
|\max D-1|=\operatorname{dist}(1,D).
\]
Each condition defines a closed family in the Hausdorff metric. For the
first and third conditions, this follows from continuity of
$D\mapsto\operatorname{dist}(a_0,D)$, $D\mapsto\max D$, and
$D\mapsto\operatorname{dist}(1,D)$. For the second condition,
convergence in the Hausdorff metric preserves containment in the closed
set $K\setminus(a_0,2t_0-a_0)$. Applying the Portmanteau theorem to the closed
family $\mathcal C$, with $\Psi_n(\mathcal C)=1$ for every $n$, gives
\[
\Psi(\mathcal C)\ge\limsup_{k\to\infty}\Psi_{n_k}(\mathcal C)=1.
\]
Weak convergence also preserves the expected maximum, since
$D\mapsto\max D$ is bounded and continuous on this space:
\begin{equation}
\label{eq:general-terminal-mean}
\E_\Psi\max D=m(1).
\end{equation}

To identify the moments induced by $\Psi$, consider the agent's payoff
from a delegation set $D$. Define
\[
\phi_D(t)=\max_{a\in D}(2at-a^2)
=t^2-\operatorname{dist}(t,D)^2.
\]
Then $\phi_D(t)-t^2$ is the agent's payoff when his ideal action is $t$.
As a maximum of affine functions with slopes $2a$, $a\in D$,
$\phi_D$ is convex, with left and right derivatives
$2a_D^-(t)$ and $2a_D^+(t)$, where $a_D^-(t)$ and $a_D^+(t)$ are the
smallest and largest best responses. These slopes are uniformly bounded
because $K$ is compact. Using the initial choice $a_0$ at $t_0$,
integrating the slopes and taking expectations gives
\[
\E_{\Psi_n}\phi_D(t)
=2t_0a_0-a_0^2+2\int_{t_0}^t m_n(v)\,dv.
\]
For every fixed $t$, boundedness and continuity of
$D\mapsto\operatorname{dist}(t,D)$ imply the same properties for
$D\mapsto\phi_D(t)$. Weak convergence of $\Psi_{n_k}$, uniform
convergence of $m_{n_k}$, and~\eqref{eq:general-ideal-action-envelope} give
\begin{equation}
\label{eq:general-limit-utility}
\Phi(t):=\E_\Psi\phi_D(t)=t^2+u^*(t),
\qquad t\in[t_0,t_1].
\end{equation}
This equality identifies every type's expected payoff under $\Psi$ with
that of the optimal mechanism. Differentiating the envelope representation
gives $\Phi'(t)=2m(t)$, with one-sided derivatives at the endpoints.

We now recover the induced mean. The uniform bound on slopes allows us
to pass the one-sided derivatives through the expectation by bounded
convergence. At every interior $t$, this gives
\[
\E_\Psi a_D^-(t)=\frac{\Phi'_-(t)}2
=m(t)=\frac{\Phi'_+(t)}2=\E_\Psi a_D^+(t).
\]
At $t_0$, almost every $D$ contains $a_0$, lies in $[a_0,\infty)$, and
misses $(a_0,2t_0-a_0)$, which makes $a_0$ its smallest best response.
The right derivative gives
$\E_\Psi a_D^+(t_0)=\Phi'_+(t_0)/2=m(t_0)=a_0
=\E_\Psi a_D^-(t_0)$.
At $t_1\ge1$, the largest best response is $\max D$, since this action
is already a best response for the type with ideal action $1$.
By~\eqref{eq:general-terminal-mean}
and constancy of $m$ on $[1,t_1]$,
\[
\E_\Psi a_D^-(t_1)
=\Phi'_-(t_1)/2
=m(t_1)=m(1)=\E_\Psi a_D^+(t_1).
\]
Since $a_D^+(t)-a_D^-(t)\ge0$, equality of their expectations implies
that each fixed type has a unique best response with probability one.
Tie-breaking has no effect on the limiting moments, while
ties under the finite rules occur only at grid points and do not affect
the envelope integrals. For every $t$, $\E_\Psi a_D(t)=m(t)$.
The quadratic payoff identity and~\eqref{eq:general-limit-utility} give
\[
\begin{aligned}
-u^*(t)
&=\E_\Psi[(a_D(t)-t)^2]
=\Var_\Psi a_D(t)+[m(t)-t]^2,\\
\Var_\Psi a_D(t)
&=-u^*(t)-[m(t)-t]^2
=\tau^*(y_A^{-1}(t)).
\end{aligned}
\]

To prove (ii), suppose $a_0=t_0$. Appendix D.1 gives $u^*(t)=0$ for
$t\in[t_0,y_A(\underline\theta)]$.
Equation~\eqref{eq:general-limit-utility} implies
\[
\E_\Psi\operatorname{dist}(t,D)^2=0
\qquad\text{for every }t\in[t_0,y_A(\underline\theta)].
\]
Nonnegativity of squared distance and closedness of $D$ imply that each
such $t$ lies in $D$ almost surely. Taking a countable dense subset of
the interval shows that, with probability one, $D$ contains all points
of this subset and, by closedness, the entire interval
$[a_0,y_A(\underline\theta)]$. If $\underline\theta=0$, this interval
reduces to $\{a_0\}$. Property (i) follows from $\Psi(\mathcal C)=1$.
All stated restrictions hold after removing a $\Psi$-null family.

The rule $\Psi$ reproduces the optimal mean and variance at every state
and attains the optimal mechanism's worst-case regret, proving
$V_D\le V^*$. Every random delegation rule with finite worst-case regret
induces an incentive-compatible stochastic mechanism, giving the reverse
inequality $V^*\le V_D$. This proves $V_D=V^*$.
\end{proof}

\subsection{Affine bias}
\label{app:affine-bias}

\begin{proof}[Proof of Proposition~\ref{prop:affine-bias}]
Suppose $b(\theta)=b_0+\beta\theta$, where $0\leq b_0<1$ and $b_0+\beta>0$. Since $1+\beta>0$, Proposition~\ref{prop:general-optimal-mechanism} gives an optimal stochastic mechanism with the moments in~\eqref{eq:general-optimal-moments} and unrestricted minimax value
$
V^*=V_{b_0,\beta}=d(\underline\theta)^2$. 
We first determine the cutoffs and then construct a random delegation rule that implements these moments.

The upper cutoff $\overline\theta$ satisfies $\overline\theta+b(\overline\theta)=1$, where the agent's ideal action at this state equals the principal's ideal action at the highest state. Thus,
$
\overline\theta=\frac{1-b_0}{1+\beta}$. 
The function $d(\theta)$ in~\eqref{eq:general-distance-solution} specifies the mean-action distortion required to sustain constant regret under incentive compatibility. It is determined backward from $\overline\theta$, where $d(\overline\theta)=b(\overline\theta)/2$.

For $\beta>0$, if $d(\theta)=b(\theta)$ has a solution in $(0,\overline\theta)$, this solution defines the lower cutoff $\underline\theta$. At this cutoff, the mean distortion equals the agent's bias, so the mean action coincides with the agent's ideal action. Otherwise, set $\underline\theta=0$. For $\beta\leq0$, set $\underline\theta=0$.

For $\beta>0$, solving~\eqref{eq:general-distance-ode} with $d(0)=b_0$ gives the threshold at which this crossing reaches state zero:
\[
b_c(\beta)=
\begin{cases}
\displaystyle
\frac{\beta}{
(1+\beta)
\left(\frac{1+\beta}{2\beta}\right)^{\beta/(1-\beta)}-1},
&\beta\ne1,\\[3mm]
\displaystyle
\frac{1}{2e^{1/2}-1},
&\beta=1.
\end{cases}
\]
For $\beta\leq0$, define $b_c(\beta)=0$. Accordingly,
\[
\underline\theta=
\begin{cases}
\displaystyle
\max\left\{0,
\frac{b_c(\beta)-b_0}{b_c(\beta)+\beta}
\right\},
&\beta>0,\\[2mm]
0,&\beta\leq0.
\end{cases}
\]

Define the response to the agent's ideal action by
$
\lambda\bigl(\theta+b(\theta)\bigr)
=\frac{d(\theta)}{(1+\beta)b(\theta)},
\forall 
\underline\theta\leq\theta\leq\overline\theta$, 
which is positive and decreasing. 
We first determine when random caps can implement the optimal moments. This requires the deterministic action to equal $b_0$, the agent's optimal action at state $0$. If that action lies strictly below $b_0$, every realized delegation set must preserve a gap immediately above it. By the characterization above, the bottom action equals $b_0$ exactly when $\beta>0$ and $b_0\leq b_c(\beta)$. 

We now verify sufficiency by constructing a cap distribution under this condition.
Consider a random delegation rule that offers $[b_0,\bar a]$, where the cap $\bar a$ is supported on
$[\underline\theta+b(\underline\theta),1]$ and satisfies
$
\Pr\bigl(\bar a>\theta+b(\theta)\bigr)
=\lambda\bigl(\theta+b(\theta)\bigr),
\forall 
\underline\theta<\theta<\overline\theta$. 

Since $\lambda(y_A(\underline\theta))=1/(1+\beta)\leq1$, monotonicity makes this a valid cap distribution. Below $\underline\theta$, every realized interval contains
the agent's ideal action. On the responsive region,
\[
\frac{\mathrm d}{\mathrm d\theta}
\E\!\left[\min\{\theta+b(\theta),\bar a\}\right]
=(1+\beta)\lambda\bigl(\theta+b(\theta)\bigr)
=\mu^{*\prime}(\theta).
\]
At $\underline\theta$, the induced action is deterministic
and equals $\mu^*(\underline\theta)$, so the mean matches
throughout the responsive region. The common initial variance
zero and the incentive-compatibility envelope then match the
variance. Above $\overline\theta$, every type chooses the
realized cap, yielding the required pooling moments.

In the remaining cases, the deterministic action equals $a_0=d(0)<b_0$. Apply the endpoint construction in Appendix~\ref{app:endpoint-construction}, with $b_0$ in place of $b$, to match the current $\lambda$. Since $\lambda$ is still decreasing, the previous construction remains feasible.

It remains to verify that the total mass assigned to
non-singleton delegation sets is less than one. At the
stopping threshold $c$, every floor satisfies
$\underline a\leq\underline a_{\max}=2c-a_0$, and every
ceiling satisfies $\bar a\geq c$ almost surely. Since $G$
has no atoms, ties at the switching threshold have zero
probability. An agent with ideal action $c$ therefore
chooses $\max\{c,\underline a\}$ from each non-singleton
set. By the mean-matching calculation,
\[
\mu^*\!\left(\frac{c-b_0}{1+\beta}\right)-a_0
=
q(c-a_0)
+\int_c^{\underline a_{\max}}
(x-c)\,\mathrm dG(x)
>
q(c-a_0).
\]
The strict inequality follows because
$\underline a_{\max}>c$ and $G$ has positive density
immediately below $\underline a_{\max}$.
At state $\theta=(c-b_0)/(1+\beta)$, the agent's ideal
action is $c$. Since the target mean is weakly below
the agent's ideal action,
\[
q(c-a_0)
<
\mu^*\!\left(\frac{c-b_0}{1+\beta}\right)-a_0
\leq c-a_0.
\]
Because $c>a_0$, this proves $q<1$.

The endpoint bounds and all remaining steps of the earlier construction apply
unchanged. The endpoint distributions reproduce the target
mean, and the common initial variance zero together with
the incentive-compatibility envelope reproduces the target
variance. Thus the default, gap, and random interval
implement the optimal moments.
\end{proof}

\bibliography{bibfile}

\clearpage
\setcounter{section}{0}
\setcounter{subsection}{0}
\pagenumbering{arabic}

\renewcommand{\thesection}{OA.\arabic{section}}
\renewcommand{\thesubsection}{\thesection.\arabic{subsection}}
\renewcommand{\thepage}{OA-\arabic{page}}

\begin{center}{\LARGE\bfseries Online Appendix}\\[0.5em]\end{center}
\section{Tie-breaking Rule}\label{subsec:tie}

In this section, we show that the optimal mechanism and random delegation rule remain optimal under any tie-breaking rule $\eta$.

First, we will prove that $M^*$ is an optimal mechanism under any tie-breaking rule. Suppose the $\mu$-allocation is continuous (and hence $\tau$ is also continuous). Then, by the envelope condition, we have
\[ U(\theta)-U(\hat{\theta} \mid \theta) = 2 \int_{\hat{\theta}}^\theta y_A^\prime(x)(\mu(x)-\mu(\hat{\theta})) \, dx. \]
If the agent is indifferent between reporting $\theta$ and $\hat{\theta}\neq\theta$ after the state $\theta$ is realized, this yields
\[ 2 \int_{\hat{\theta}}^\theta y_A^\prime(x)(\mu(x)-\mu(\hat{\theta})) \, dx=0. \]
Since $y_A$ is strictly increasing and $\mu$ must be weakly increasing by IC, this holds if and only if $\mu(\cdot)$ is constant between the two reports. By the envelope condition, $\tau(\cdot)$ needs to be constant on the same region as well. Consequently, the principal will also be indifferent among all reports that are optimal for the agent as the allocation $(\mu,\tau)$ remains identical across these optimal reports. Hence, all equilibrium selection rules are payoff-equivalent for any incentive-compatible mechanism with a continuous mean allocation. 

For any incentive-compatible mechanism $M$ and a tie-breaking rule $\eta$, we can establish the following relation:
\[ \sup_{\theta \in \Theta} R(M,\theta,\eta) \ge\inf_{\eta^\prime} \sup_{\theta \in \Theta} R(M,\theta,\eta^\prime)\ge \inf_{\eta^\prime} \sup_{\theta \in \Theta} R(M^*,\theta,\eta^\prime)=\sup_{\theta \in \Theta} R(M^*,\theta,\eta). \]
The first inequality follows from the definition of the infimum. The second follows because every admissible selection among the agent's optimal reports induces an incentive-compatible allocation, whose worst-case regret is at least $V^*$, the value attained by $M^*$. The final equality follows from the preceding continuity argument: all optimal reports under $M^*$ induce the same mean and variance, so its regret is unchanged among all tie-breaking rules.

Next, we consider the implementation of these stochastic mechanisms via lotteries over delegation sets. With a slight abuse of notation, let $\eta$ again define the agent's tie-breaking rule. If there exists no state at which the agent is indifferent between two actions within the realized delegation set given the optimal lottery with positive probability, then every measurable tie-breaking rule induces the same action distribution at every state. Therefore, the regret the principal attains against any tie-breaking rule is identical to the regret attained under principal-favoring tie-breaking.
If $y_A$ is affine and we are in case (i) of Proposition~\ref{prop:affine-bias}, we know that the floor distribution $G$ is continuous and hence there is a zero probability that the agent is indifferent between two actions at every state $\theta$ under the lottery $\Psi^*$. If we are in case (ii) of Proposition~\ref{prop:affine-bias}, then there exists no state at which the agent is indifferent between two actions from every realized delegation set under $\Psi^*$ (as every set is an interval).

Now, consider the more general form of $y_A$ we allow for in Proposition~\ref{prop:general-state-dependent-bias}. In what follows, we show that it is also true here that, for every $\theta \in [0,1]$,
\[ \mathbb{P}_{\Psi^*}\left( \left\{ D: |\arg\min_{a \in D} \, (a-y_A(\theta))^2|>1 \right\} \right)=0. \]
So, as before, we want to show that there exists no state at which under the limiting lottery $\Psi^*$ the agent is indifferent with positive probability.

We show that the best response is unique almost surely at each fixed type.
Let $a_D^-(\theta)$ and $a_D^+(\theta)$ denote the actions selected from $D$ at state $\theta$ under the tie-breaking rules that choose the lowest and highest best responses, respectively. For now, fix $\theta<1$. For every delegation set $D$, for every $\hat{\theta}\in(\theta,1]$
\[
    a_D^-(\hat{\theta})\geq a_D^+(\theta)\geq a_D^-(\theta),
\]
as at state $\theta$, $a_D^+(\theta)$ must be preferred over $a_D^-(\hat{\theta})$ and vice versa. Adding these inequalities, we obtain:
\[
    (y_A(\hat{\theta})-y_A(\theta))\left(a_D^-(\hat{\theta})-a_D^+(\theta)\right)\geq 0.
\]
Because $\Psi^*$ implements $\mu^*$ under lower-action tie-breaking, taking expectations yields
\[ 0 \leq \int \left( a_D^+(\theta)-a_D^-(\theta)  \right) \, d\Psi^*(D) \leq \mu^*(\hat\theta)-\mu^*(\theta).\]
Letting $\hat{\theta}\searrow \theta$ and using continuity of $\mu^*$, we see that the right-hand side converges to zero. As the integrand is nonnegative, we obtain
\[a_D^+(\theta)=a_D^-(\theta) \quad \Psi^*\text{-almost surely}.\]
This proves uniqueness at every fixed $\theta<1$, including $\theta=0$.

For the highest type $\theta=1$, the finite-grid construction gives
\[ \int \max D \, d \Psi_n(D)=\mu^*(1). \]
Since the mapping $D \mapsto \max D$ is bounded and continuous in the Hausdorff metric, weak convergence implies
\[\int \max D \, d\Psi^*(D)=\mu^*(1).\]
As $\Psi^*$ implements the target means under lower-action tie-breaking, this yields
\[ \int a_D^-(1) \, d\Psi^*(D)=\mu^*(1).\]
Taking these identities together, we arrive at
\[ a_D^-(1)=\max D \quad \Psi^*\text{-almost surely}.\]
Consequently, $a_D^-(1)$ must be the unique best response at $\theta=1$ almost surely.
Therefore, for every fixed $\theta\in[0,1]$, we indeed have
\[ \mathbb{P}_{\Psi^*}\left( \left\{ D: |\arg\min_{a \in D} \, (a-y_A(\theta))^2|>1 \right\} \right)=0. \]

This implies that for all tie-breaking rules $\eta$
\[ \inf_{\eta^\prime} \sup_{\theta \in \Theta} R(\Psi^*,\theta,\eta^\prime)=\sup_{\theta \in \Theta} R(\Psi^*,\theta,\eta).\]

Therefore, for any lottery $\Psi$ over delegation sets and any admissible
tie-breaking rule $\eta$, we can establish
\[
\sup_{\theta\in\Theta}R(\Psi,\theta,\eta)
\ge \inf_{\eta'}\sup_{\theta\in\Theta}R(\Psi,\theta,\eta')
\ge \inf_{\eta'}\sup_{\theta\in\Theta}R(\Psi^*,\theta,\eta')
= \sup_{\theta\in\Theta}R(\Psi^*,\theta,\eta),
\]
where the first inequality follows from the definition of the infimum and the second follows because principal-favoring tie-breaking minimizes
regret at every state, and $\Psi^*$ is optimal under this convention.

\section{Exact implementation and payoff implementation}
\label{subsec:payoff-implementation}

A delegation rule is an \textit{exact implementation} of a direct mechanism if the two induce the same action distribution at every state. It is a \textit{payoff implementation} of a mechanism if both players obtain the same expected utilities at every state. In the present quadratic environment, payoff implementation is equivalent to matching the mean and the second moment. Indeed, the difference between the agent's and principal's expected losses equals $b^2-2b(\mu(\theta)-\theta)$, so the two losses identify the mean and then the second moment.

Payoff implementation itself may not be possible for an arbitrary incentive-compatible mechanism. Consider the three lotteries
\[
Q_0=\delta_b,\qquad
Q_1=\frac38\delta_b+\frac58\delta_{b+6/5},\qquad
Q_2=\delta_{b+3/2}.
\]
The agent's expected losses from these lotteries are
\[
\ell_0(\theta)=\theta^2,\qquad
\ell_1(\theta)=\theta^2-\frac32\theta+\frac9{10},\qquad
\ell_2(\theta)=\theta^2-3\theta+\frac94.
\]
A mechanism selecting $Q_0$ for $\theta\leq3/5$, $Q_1$ for $3/5<\theta\leq9/10$, and $Q_2$ for $\theta>9/10$ is therefore incentive compatible.

Suppose a random delegation rule matches its moments. The deterministic actions at states zero and one require almost every realized set to contain both $b$ and $b+3/2$. For state one to choose $b+3/2$, such a set must exclude $(b+1/2,b+3/2)$. State $7/10$ strictly prefers the available action $b$ to any action at or above $b+3/2$. Its delegated choice must consequently be at most $b+1/2$, but $Q_1$ has mean $b+3/4$, a contradiction. The distinction is that a direct mechanism can offer $b+6/5$ only as part of a lottery, whereas a delegation set would let the highest state select it without taking the associated risk.

\end{document}